\documentclass[10pt,conference]{IEEEtran}

\usepackage{graphicx}
\graphicspath{{./figs/}}
\usepackage{algorithm}
\usepackage{algorithmic}
\usepackage{amsmath}
\usepackage{subfigure}
\usepackage{booktabs}

\begin{document}

\title{\huge Layout Decomposition for Triple Patterning Lithography}

\author{
\IEEEauthorblockN{Bei Yu, Kun Yuan\IEEEauthorrefmark{2}, Boyang Zhang, Duo Ding, David Z. Pan}
\IEEEauthorblockA{ECE Dept. University of Texas at Austin, Austin, TX USA 78712 \\}
\IEEEauthorblockA{\IEEEauthorrefmark{2}	Cadence Design Systems, Inc., San Jose, CA USA 95134\\
Email: \{bei, dpan\}@cerc.utexas.edu
}
}

\maketitle

\newtheorem{problem}{\textbf{Problem}}
\newtheorem{define}{\textbf{Definition}}
\newtheorem{theorem}{\textbf{Theorem}}
\newtheorem{lemma}{\textbf{Lemma}}
\newtheorem{conjecture}{Conjecture}

\begin{abstract}
As minimum feature size and pitch spacing further decrease, triple patterning lithography (TPL) is a possible 193nm extension along the paradigm of double patterning lithography (DPL). However, there is very little study on TPL layout decomposition.  In this paper, we show that TPL layout decomposition is a more difficult problem than that for DPL. We then propose a general integer linear programming formulation for TPL layout decomposition which can simultaneously minimize conflict and stitch numbers.
Since ILP has very poor scalability,  we propose three acceleration techniques without sacrificing solution quality:
independent component computation, layout graph simplification, and bridge computation. 
For very dense layouts, even with these speedup techniques, ILP formulation may still be too slow. Therefore, we propose a novel vector programming formulation for TPL decomposition, and solve it through effective semidefinite programming (SDP) approximation.
Experimental results show that the ILP with acceleration techniques can reduce 82\% runtime compared to the baseline ILP.
Using SDP based algorithm, the runtime can be further reduced by 42\% with some tradeoff in the stitch number (reduced by 7\%) and the conflict (9\% more).  However, for very dense layouts, SDP based algorithm can achieve $140\times$ speed-up even compared with accelerated ILP.
\end{abstract}

\section{Introduction}

As minimum feature size further scales, the semiconductor industry is greatly challenged of patterning sub-22nm half-pitch due to the delay of viable next generation lithography such as Extreme Ultra Violet (EUV).
Double patterning lithography (DPL) is widely recognized as a promising solution for 32nm, 22nm, and possibly 16nm volume chip production.

As shown in Fig. \ref{fig:DPL}, the key challenge of DPL lies in the decomposition process by which the original layout is divided into two masks.
Then, there are two exposure/etching steps, through which the layout can be produced.
The advantage of this approach is that the effective pitch can be doubled, which improves the lithography resolution.
During the decomposition, when the distance between the two patterns is less than minimum colorable distance $min_s$, they need to be assigned to different masks to avoid a conflict.
Sometimes conflict can be solved by splitting a pattern into two touching parts, called stitches.
In Fig. \ref{fig:DPL}(b), polygon $a$ is split into two polygons $a_1$ and $a_2$ in order to resolve the decomposition conflicts.
However the introduced stitches lead to yield loss due to overlay error  \cite{DPL_ICCAD08_Yang}.
Therefore, two of the main challenges in layout decomposition are conflict and stitch minimization.

The paradigm of double patterning may be further extended to triple patterning lithography (TPL).
Industry has already explored the test-chip patterns with triple patterning or even quadruple patterning \cite{2009Intel}.
By using TPL, we can achieve further feature-size scaling through pitch-tripling.

It shall be noted that in DPL, even with stitch insertion, there may be native conflicts \cite{DPL_SPIE07_Anton}.
Fig. \ref{fig:TPL}(a) shows a three-way conflict cycle between features a, b and c, where any two of them are within the $min_s$.
As a consequence, there is no chance to produce a conflict-free solution using DPL decomposition.
However, we can easily resolve this problem if the layout is decomposed into three masks as shown in Fig. \ref{fig:TPL}(b). 
Yet this does not mean TPL layout decomposition problem becomes easier. Actually since the features can be packed closer, the problem turns out to be more difficult, as to be shown in Section  \ref{chap:problem}.

\begin{figure}[bt]
    \centering
    \subfigure[]{\includegraphics[width=0.15\textwidth]{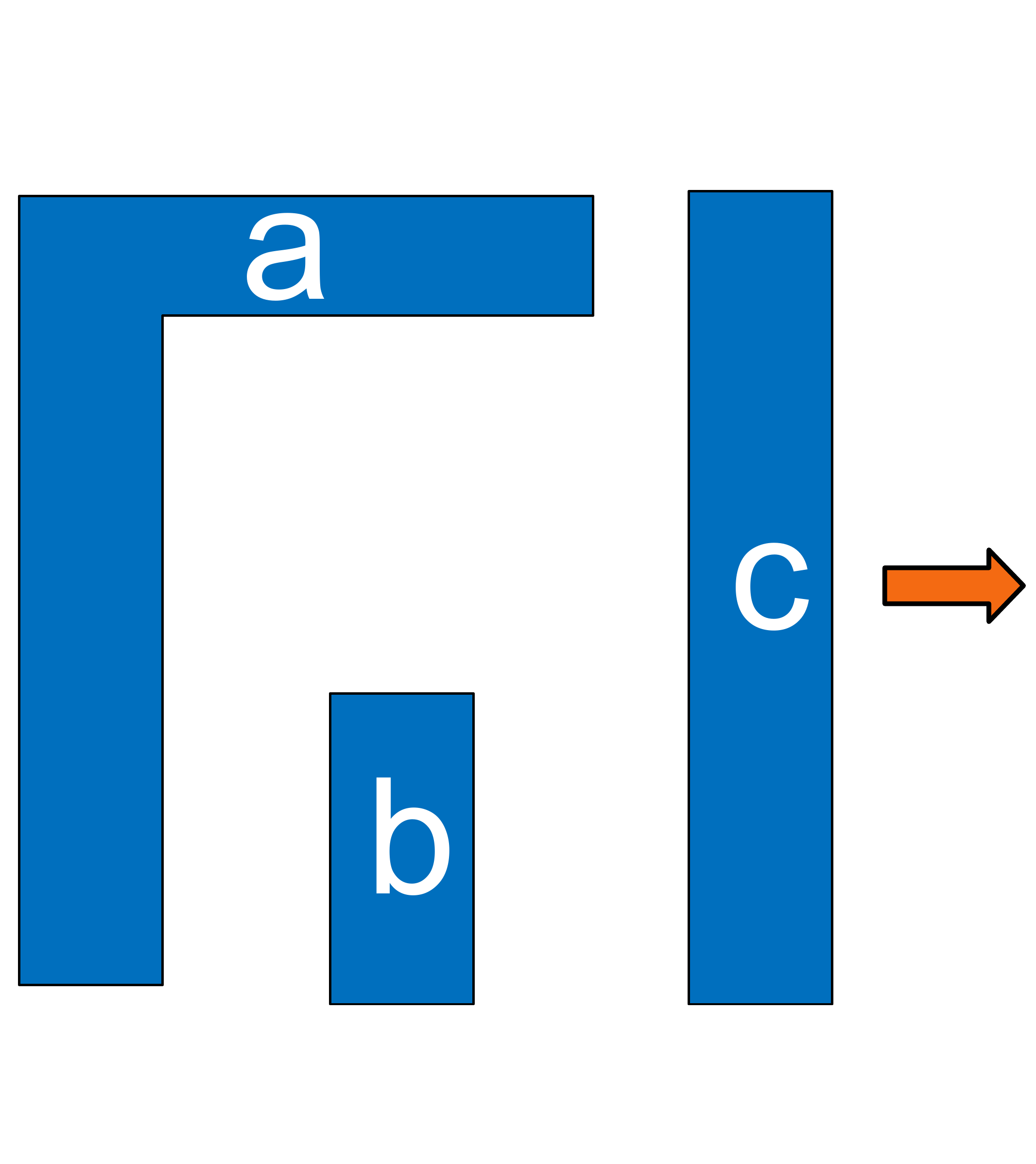}}
    \subfigure[]{\includegraphics[width=0.15\textwidth]{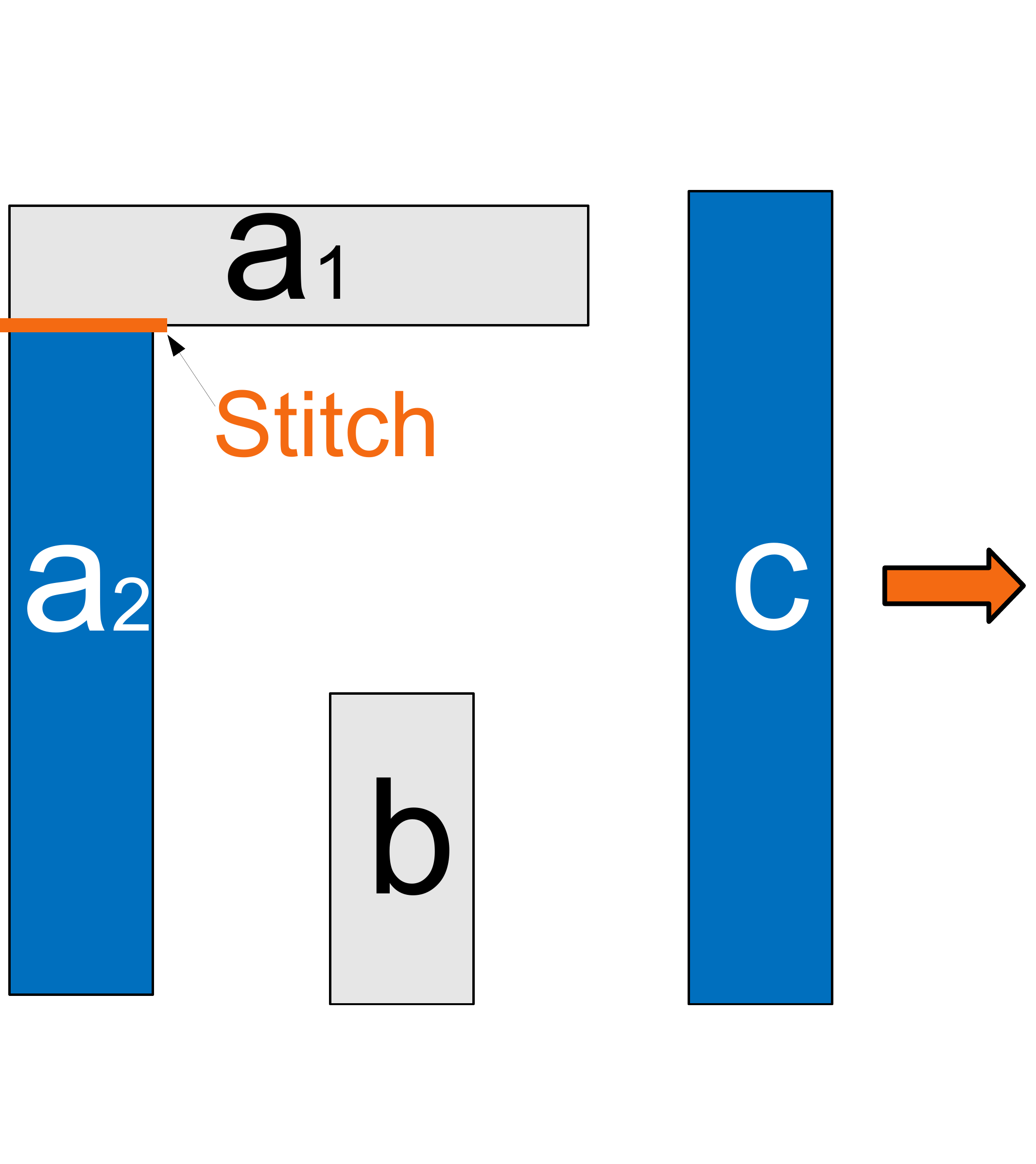}}
    \subfigure[]{\includegraphics[width=0.13\textwidth]{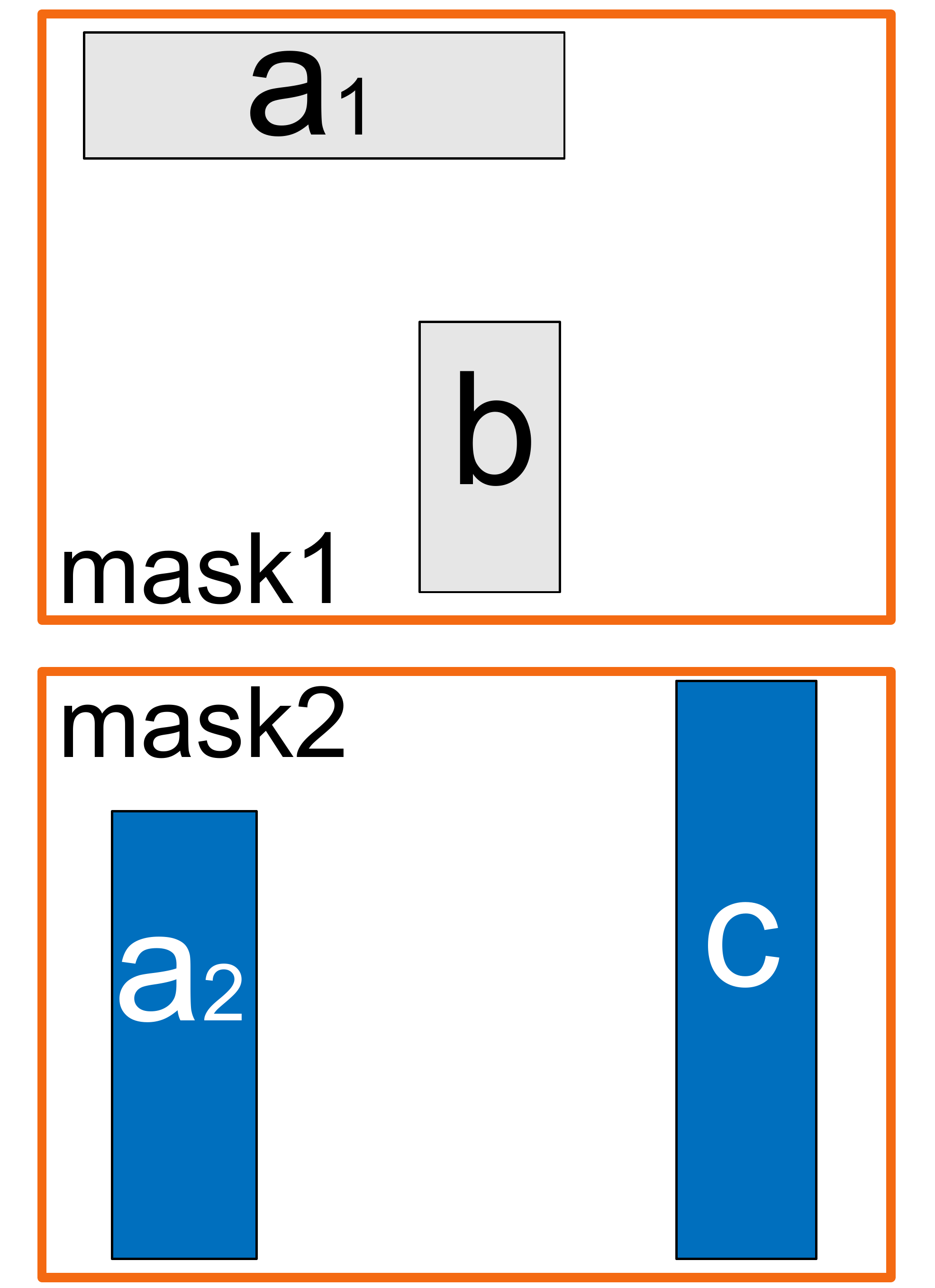}}
    \caption{~In DPL, a single layer is decomposed into two masks and the pitch can be increased effectively.}
    \label{fig:DPL}
\end{figure}

\begin{figure}[bt]
    \centering
    \subfigure[]{\includegraphics[width=0.22\textwidth]{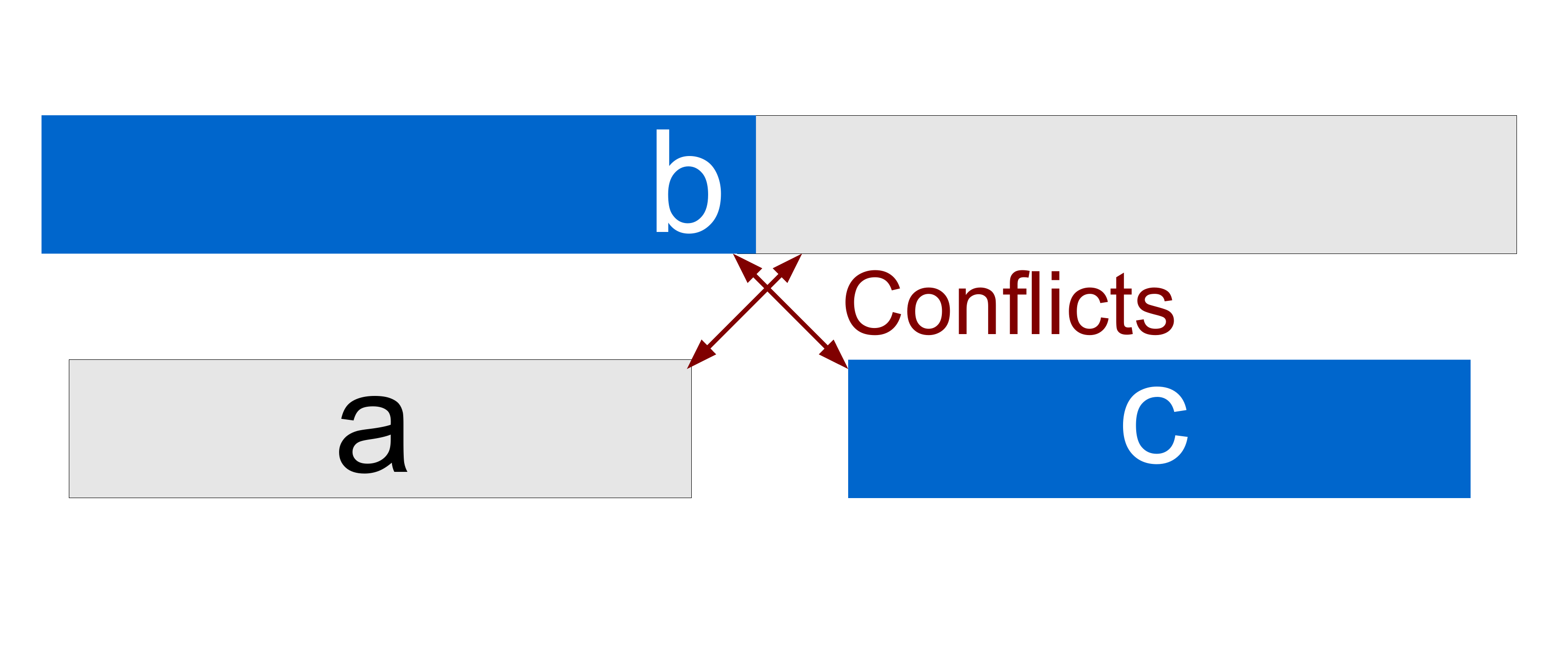}}
    \subfigure[]{\includegraphics[width=0.22\textwidth]{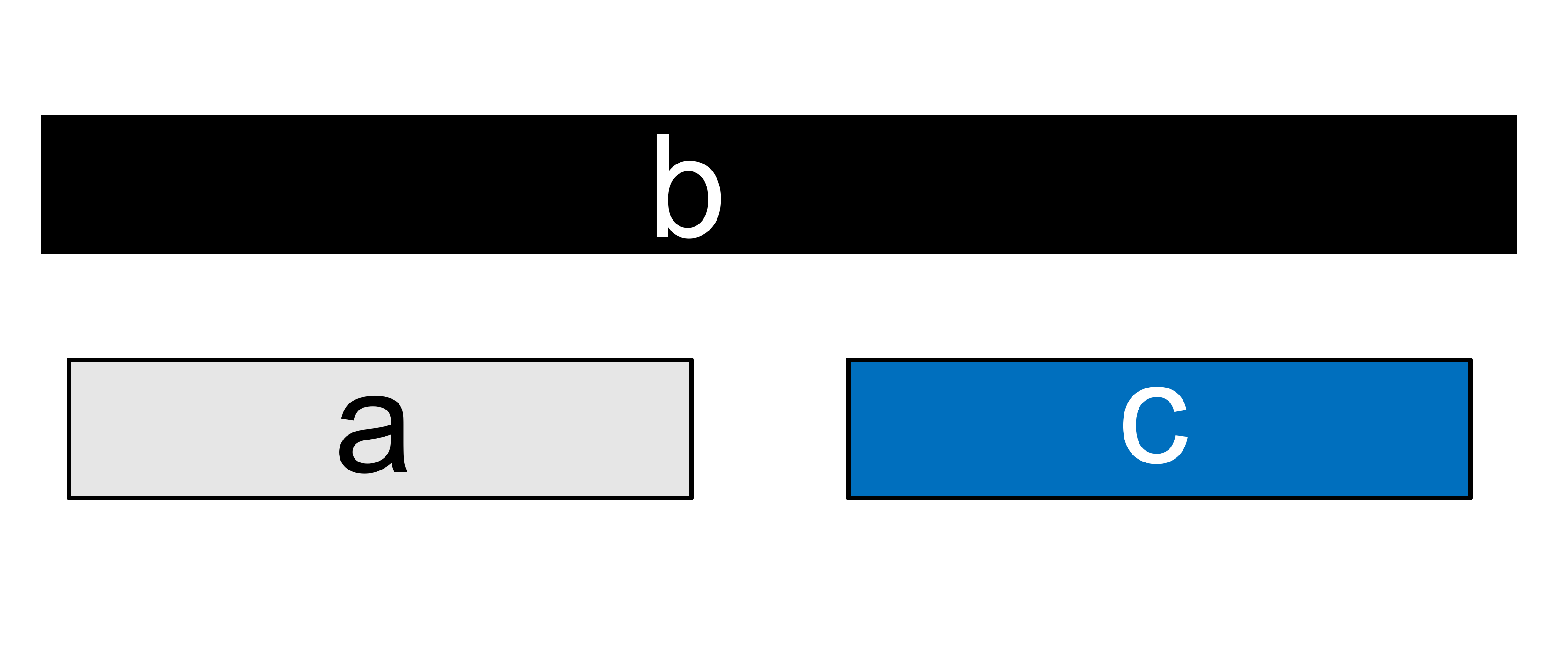}}
    \caption{~(a) In DPL, even stitch insertion can not avoid native conflict.~(b) The native conflicts can be resolved by TPL.}
    \label{fig:TPL}
\end{figure}

Much previous research focuses on the double patterning layout decomposition problem, which is generally regarded as a two-coloring problem on a conflict graph.
Integer Linear Programming (ILP) is adopted in \cite{DPL_ICCAD08_Kahng}\cite{DPL_ISPD09_Yuan} to minimize the stitch number and/or the conflict number.
Xu et al. \cite{DPL_ICCAD09_Xu} propose an efficient graph reduction-based algorithm for stitch minimization,
and Yang et al. \cite{DPL_ASPDAC2010_Yang} propose a fast min-cut based approach.
A matching based decomposer is proposed to minimize both the conflict and the stitch numbers \cite{DPL_ISPD2010_Xu}.
To resolve the native conflict, several works introduce layout modification to further minimize the conflict number \cite{DPL_ICCAD2010_Yuan}\cite{DPL_ICCAD2010_Chen}\cite{DPL_ICCAD09_Hsu}.
However, layout modification may cause new problems, i.e., timing closure, hotspot.

Until now, there are very few investigations on TPL layout decomposition.
\cite{TPL_SPIE08_Cork} proposes a triple patterning coloring algorithm, which adopts SAT Solver. However, their work only deals with contact arrays, not general layout structures with wires, contacts, and so on. Besides, it does not involve any stitch minimization.
\cite{TPL_SPIE2011_Chen}\cite{TPL_SPIE2011_Mebarki} propose a self-aligned triple patterning (SATP) process to extend 193nm immersion lithography to half-pitch 15nm patterning.
But the SATP process cannot insert any stitch, which would greatly constrain the possible layout patterns that are decomposable \cite{SADP_2010_Li}.  To our best knowledge, there is no study so far on layout decomposition on TPL for general layout styles. 

In this paper, we propose the first systematic study on layout decomposition for triple patterning lithography.  We first formulate a general ILP formulation for TPL layout decomposition to simultaneously minimize conflict and stitch.
To improve scalability, we further propose three acceleration techniques without loss of solution quality:
layout graph simplification, independent component computation and bridges computation.
A semidefinite programming based approximation algorithm is further proposed to improve scalability.
Semidefinite programming is an extension of linear programming to approximately solve NP-hard problems and it has been successfully applied to many combinatorial problems \cite{SDP_1996Boyed}\cite{1998Karger}. 
Our main contributions of this paper include:
\begin{itemize}
    \item General ILP formulation to simultaneously minimize conflict and stitch for TPL layout decomposition;
    \item Three acceleration techniques to improve ILP scalability;
    \item A novel vector programming formulation for TPL decomposition and its semidefinite programming based approximation algorithm which can further deal with very dense layouts where even accelerated ILP becomes too slow;
    \item Our experimental results are very promising in terms of quality of results and runtime tradeoff. 
\end{itemize}

The rest of the paper is organized as follows:
in Section \ref{chap:problem}, we discuss the problem formulation and then analyze problem complexity.
The basic algorithm and some acceleration techniques are described in section \ref{chap:basic}.
Section \ref{chap:SDP} proposes a semidefinite programming based algorithm to further accelerate the basic algorithm.
Section \ref{chap:result} presents the experiment results, followed by conclusion in Section \ref{chap:conclusion}.

\section{Problem Formulation and Complexity}
\label{chap:problem}

Some preliminaries on TPL are provided in this section, including some definitions and the problem formulation.
We also demonstrate the complexity of the problem.

\subsection{Problem Formulation}

\begin{figure}[bt]
	\centering
	\subfigure[]{\includegraphics[width=0.18\textwidth]{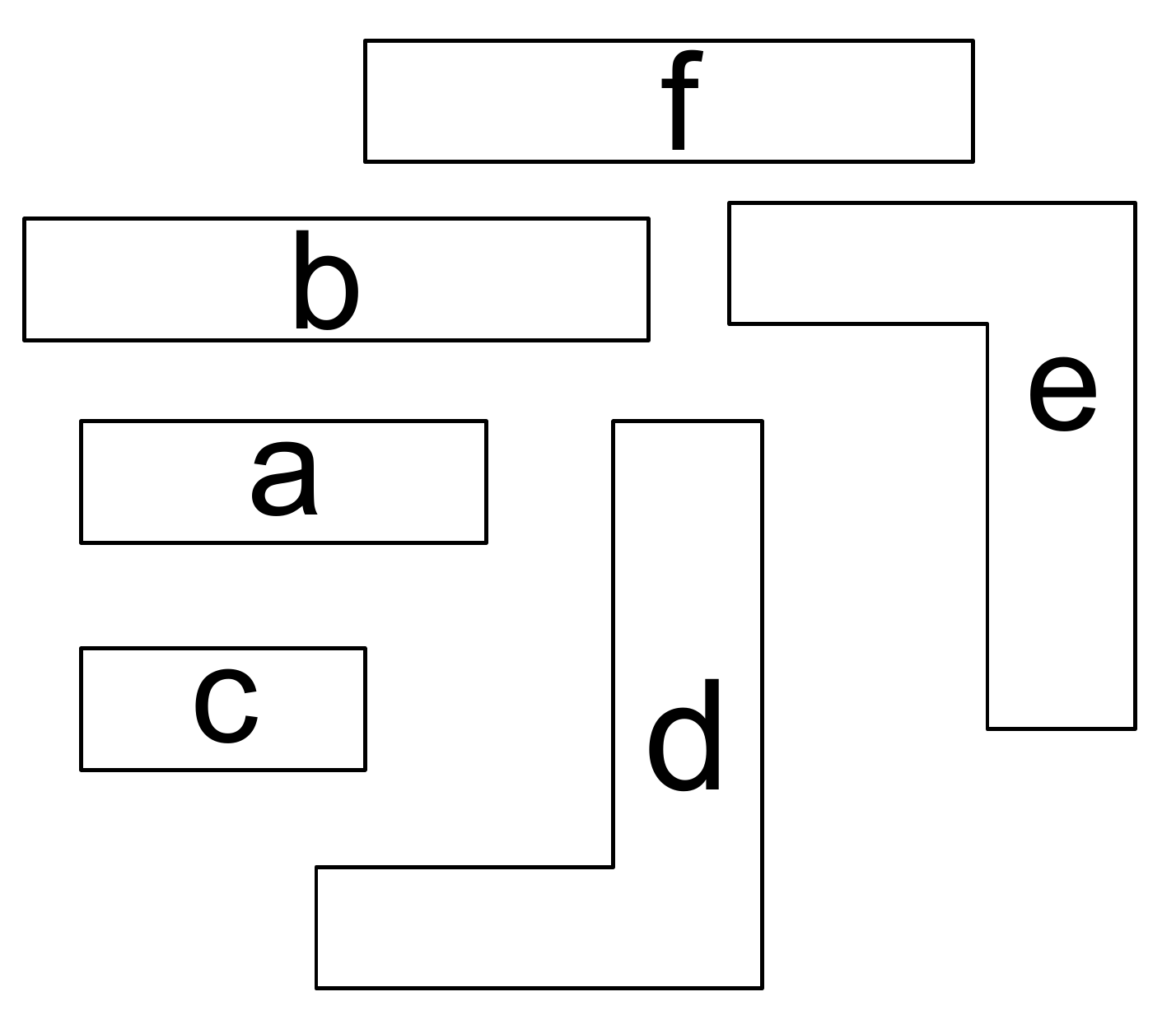}}
	\subfigure[]{\includegraphics[width=0.18\textwidth]{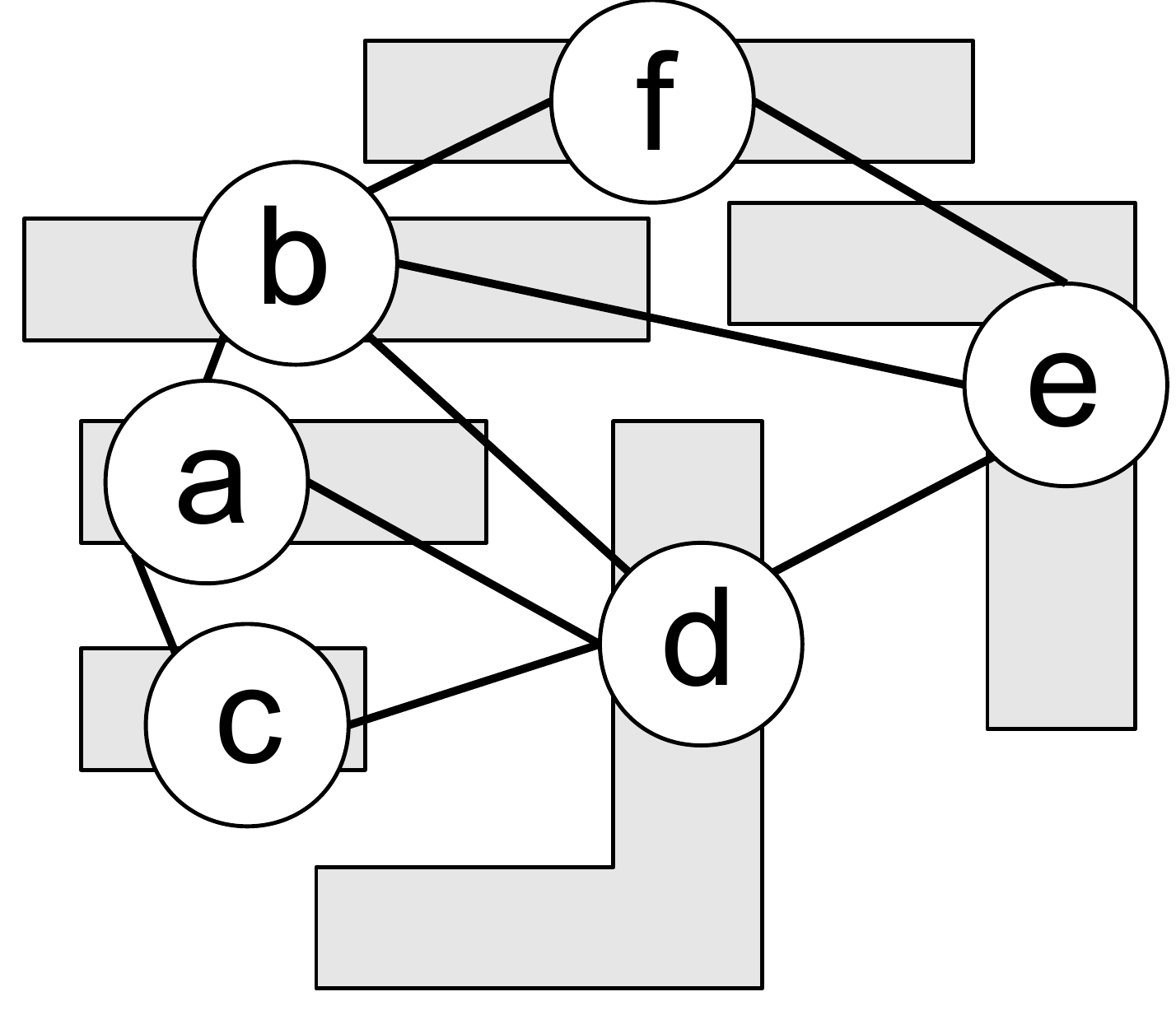}}
	\hspace{.1em}
	\subfigure[]{\includegraphics[width=0.18\textwidth]{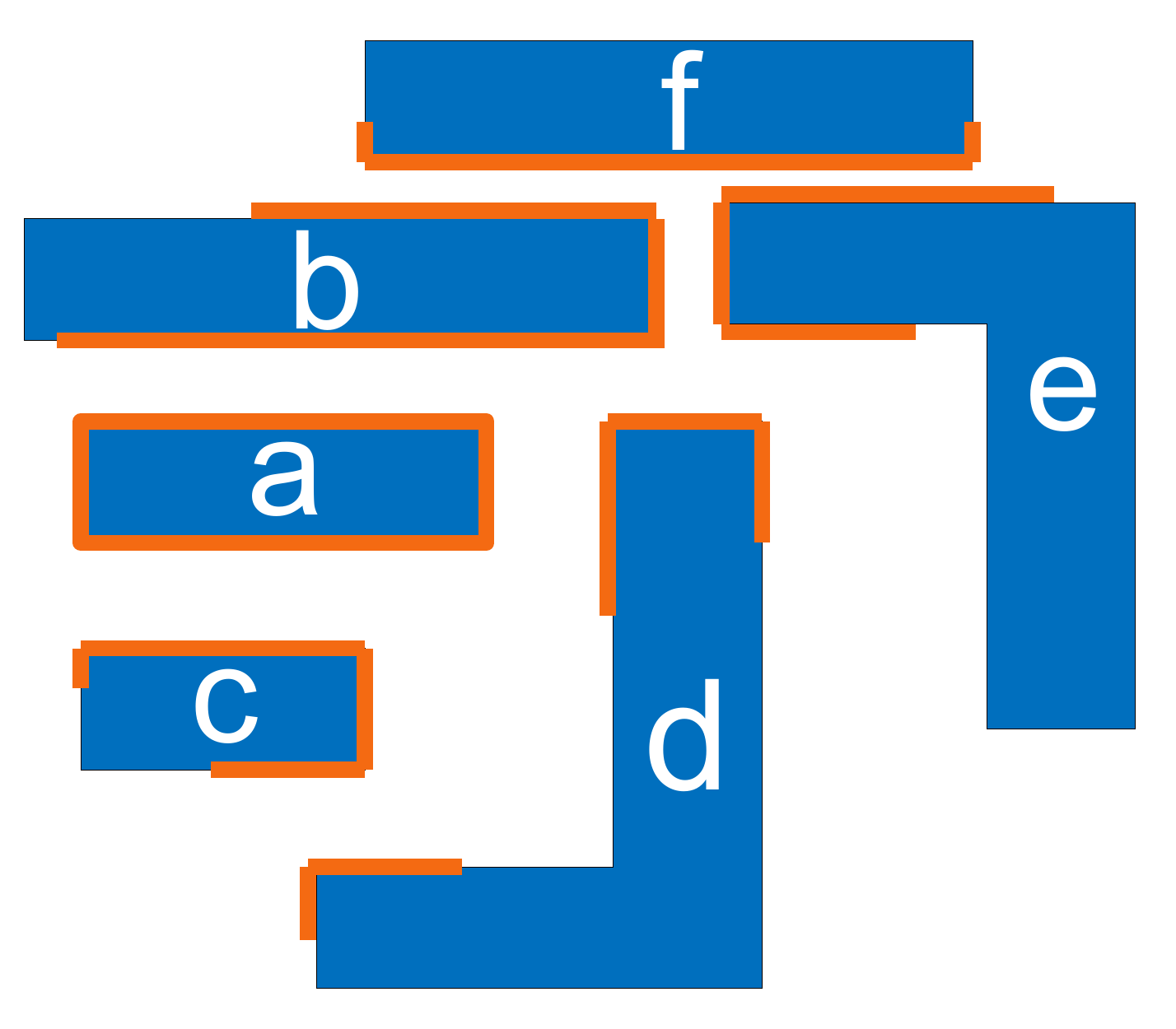}}
	\subfigure[]{\includegraphics[width=0.18\textwidth]{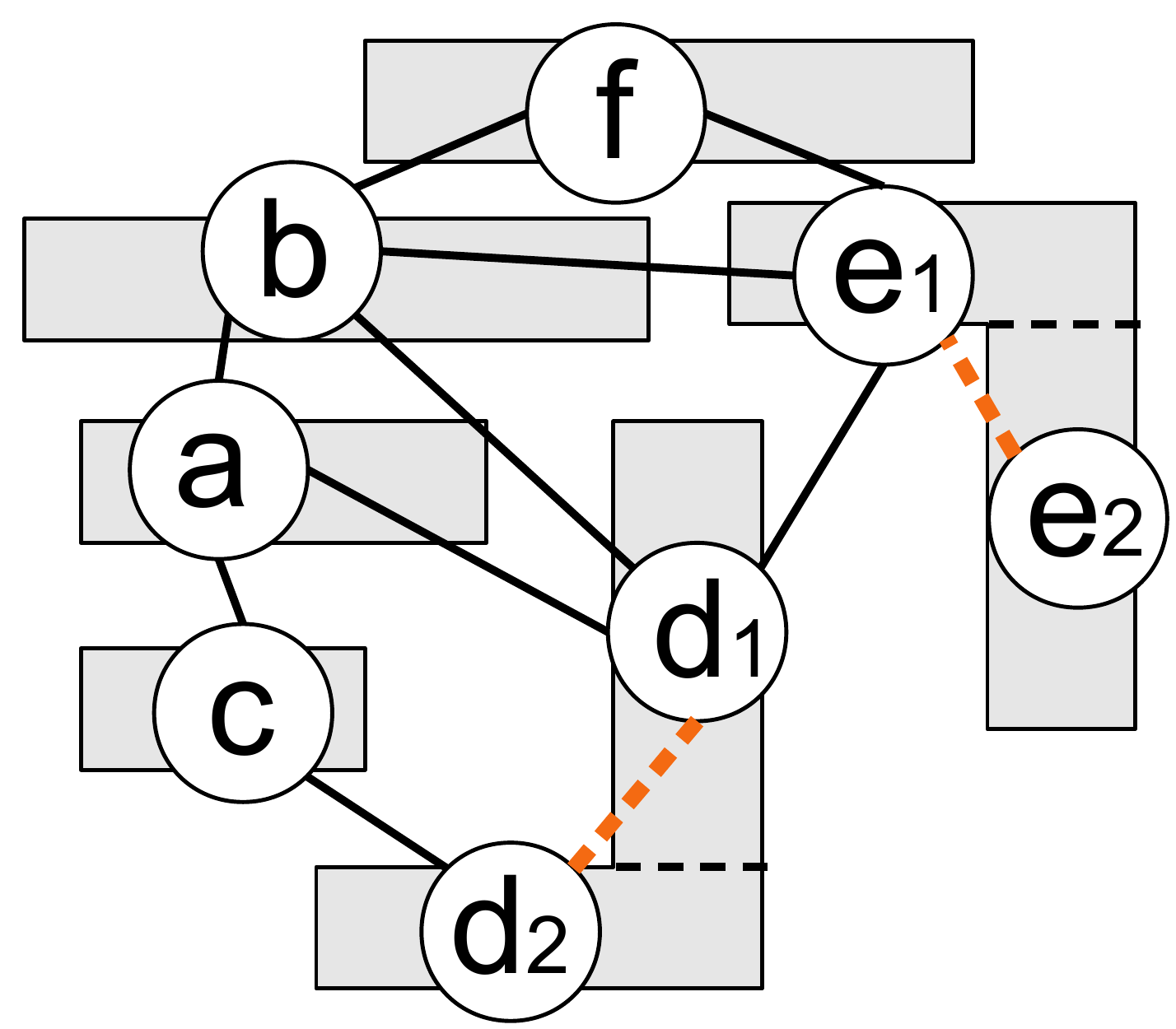}}
	\caption{~Layout graph construction and decomposition graph construction~(a) Input layout representing as irregular polygons.~(b)Corresponding layout graph, where all edges are conflict edges.~(c)The node projection.~(d)Corresponding decomposition graph, where dash edges are stitch edges.}
	\label{fig:input}
\end{figure}

Given a layout which is specified by features in polygonal shapes, a layout graph \cite{DPL_ICCAD08_Kahng} and a decomposition graph \cite{DPL_ICCAD2010_Yuan}  are constructed.


\begin{define}[\textbf{Layout Graph}]
The \textit{ layout graph} (LG) is an undirected graph whose nodes are the given layout's polygonal shapes and where an edge exists if and only if the two polygonal shapes are within minimum coloring distance $min_s$ of each other.
\end{define}

Fig. \ref{fig:input}(a) gives an example of an input layout; the corresponding layout graph is shown in Fig. \ref{fig:input}(b).
All the edges in a layout graph are called Conflict Edges (CE).
A conflict exists if and only if two nodes are connected by a CE and are in the same mask.
In other words, each conflict edge is a conflict candidate.


\begin{define}[\textbf{Decomposition Graph}]
Given a layout represented by a set of polygonal shapes, the \textit{decomposition graph} (DG) is an undirected graph with a single set of nodes $V$, and two sets of edges, $CE$ and $SE$, which contain the \textit{conflict edges} and \textit{stitch edges}, respectively.
$V$ has one or more nodes for each polygonal shape and each node is associated with a polygonal shape.
An edge is in $CE$ iff the two polygonal shapes are within minimum coloring distance $min_s$ of each other.
An edge is in $SE$ iff there is a stitch between the two nodes which are associated with the same polygonal shape.
\end{define}

On the layout graph, the node projection is first performed, where projected segments are highlighted by bold lines in Fig. \ref{fig:input}(c).
Based on the projection result, all the legal splitting locations are computed.
Then the decomposition graph is constructed, as shown in Fig. \ref{fig:input}(d).
Note that the conflict edges are marked as black edges, while stitch edges are marked as dash edges.


\begin{problem}[\textbf{TPL layout decomposition}]
Given a layout which is specified by features in polygonal shapes, the layout graph and the decomposition graph are constructed.
Our goal is to assign all the nodes in the decomposition graph to three masks to minimize the stitch number and the conflict number.
\end{problem}

\subsection{Problem Complexity}

At first glance, the layout decomposition is similar to graph coloring problem.
However, since stitch edges are introduced, the problem to minimize conflict and stitch is more complicated.
For double patterning case, deciding whether a graph is 2-colorable is easy by determining if there exists odd cycles.
For conflict and stitch minimization, if a decomposition graph is planar DPL layout decomposition can be solved in polynomial time \cite{DPL_ISPD2010_Xu}.
In order to solve the triple patterning issue, the problem becomes more complicated.

\begin{lemma}
\label{lem:dp3}
Deciding whether a \textit{planar} graph is 3-colorable is NP-complete \cite{1976Garey}.
\end{lemma}


Lemma \ref{lem:dp3} can be naturally extended to general graph.
Based on Lemma \ref{lem:dp3}, the methodology in \cite{TPL_SPIE08_Cork} is not suitable for TPL decomposition:
SAT solver can only work for 3-colorable layout graph, which cannot be checked in polynomial time.

\begin{lemma}
Coloring a 3-colorable graph with 4 colors is NP-complete \cite{1992Khanna}.
\label{lem:assign4color}
\end{lemma}

3-coloring problem is to assign the nodes in one 3-colorable graph to 3 colors.
Since coloring the graph with 4 colors cannot be finished in polynomial time, it can be shown 3-coloring problem is NP-hard.
Based on above lemmas, even the decomposition graph is planar, we reach the following theorem:

\begin{theorem}
\textbf{TPL layout decomposition problem is NP-hard.}
\end{theorem}

We can prove this theorem by reducing 3-Coloring problem to the TPL decomposition problem.
Due to page limit, the detailed proof is skipped here.

\section{Basic Algorithm}
\label{chap:basic}

In this section, we will present our basic algorithms, which are based on the Integer Linear Programming (ILP).
Since the timing complexity for ILP is very high, we propose three acceleration techniques to divide the whole problem into several smaller ones. 
The entire flow is shown in Fig. \ref{fig:overview}.

\begin{table}[bt]
\renewcommand{\arraystretch}{1.3}
\centering 
\caption{Notation}
\label{table:ilp}
\begin{tabular}{|c|c|}
	\hline \hline
	\multicolumn{2}{|c|}{Notation used in Mathematical programming}\\
	\hline
	$CE$		& set of conflict edges\\
	\hline
	$SE$		& set of stitch edges.\\
	\hline
	$V$		& the set of polygens.\\
	\hline
	$r_i$		& the $i_{th}$ layout polygons\\
	\hline
	$x_i$		& variable denoting the coloring of $r_i$\\
	\hline
	$c_{ij}$	& 0-1 variable, $c_{ij}=1$ when a conflict between $r_i$ and $r_j$\\
	\hline
	$s_{ij}$	& 0-1 variable, $s_{ij}=1$ when a stitch between $r_i$ and $r_j$\\
	\hline
	\multicolumn{2}{|c|}{Notation used in ILP formulation}\\
	\hline
	$x_{i1}, x_{i2}$ & two 1-bit 0-1 variables to represents 3 colors\\
	\hline
	$c_{ij1}, c_{ij2}$ & two 1-bit 0-1 variables to determine $c_{ij}$\\
	\hline
	$s_{ij1}, s_{ij2}$ & two 1-bit 0-1 variables to determine $s_{ij}$\\
	\hline\hline
\end{tabular}
\end{table}

\subsection{Mathematical Formulation for TPL Decomposition}

The mathematical formulation for TPL layout decomposition is shown in (\ref{eq:math}).
For convenience, some notations in mathematical programming and ILP formulation are listed in Table \ref{table:ilp}.
The objective is to simultaneously minimize both the conflict number and the stitch number.
The parameter $\alpha$ is a user-defined parameter for assigning relative importance between the conflict number and the stitch number.

\begin{align}
    \label{eq:math}
    \textrm{min}  	&     \ \ \ \  \sum_{e_{ij} \in CE} c_{ij} + \alpha \sum_{e_{ij} \in SE}s_{ij}& \\
    \textrm{s.t}.\ \ 	& c_{ij} = ( x_i == x_j )	&\forall e_{ij} \in CE		\label{1a}\tag{$1a$}\\
            		& s_{ij} = x_i \oplus x_j	&\forall e_{ij} \in SE		\label{1b}\tag{$1b$}\\
            		& x_i \in \{0, 1, 2\}		& \forall i \in V		\label{1c}\tag{$1c$}
\end{align}
where $x_i$ is a variable for the three colors of rectangles $r_i$, $c_{ij}$ is a binary variable for conflict edge $e_{ij} \in CE$ and $s_{ij}$ is a binary variable for  stitch edge $e_{ij} \in SE$. 
Constraint (\ref{1a}) is used to evaluate the conflict number when touch nodes $r_i$ and $r_j$ are assigned different colors (masks). Constraint (\ref{1b}) is used to calculate the stitch number. If node $r_i$ and node $r_j$ are assigned the same color (mask), stitch $s_{ij}$ is introduced.

\subsection{ILP Formulation for TPL Layout Decomposition}

We will now show how to formulate (\ref{eq:math}) with Integer Linear Programming.
Note that  eqs. ($1a$) and ($1b$) can be linearized only when $x_i$ is a 0-1 variable \cite{DPL_ICCAD08_Kahng}, which cannot represent three different colors. To handle this problem, we represent the color of each node using two 1-bit 0-1 variables $x_{i1}$ and $x_{i2}$.
In order to limit the number of colors for each node to 3, for each pair $(x_{i1}, x_{i2})$ the value $(1, 1)$ is not permitted.
In other words, only values $(0, 0), (0, 1)$ and $(1, 0)$ are allowed.

Thus, (\ref{eq:math}) can be formulated as follows:
\begin{align}
	\textrm{min}  &       \sum_{e_{ij} \in CE} c_{ij} + \alpha \sum_{e_{ij} \in SE}s_{ij}& 		\label{eq:ILP}\\
	\textrm{s.t}.\ \	&	x_{i1} + x_{i2} \le 1										\label{2a}\tag{$2a$}\\
				&	x_{i1} + x_{j1} \le 1 + c_{ij1}                 		&\forall e_{ij} \in CE  	\label{2b}\tag{$2b$}\\
				&	(1-x_{i1})+(1-x_{j1}) \le 1 + c_{ij1}		&\forall e_{ij} \in CE   	\label{2c}\tag{$2c$}\\
				&	x_{i2} + x_{j2} \le 1 + c_{ij2}			&\forall e_{ij} \in CE  	\label{2d}\tag{$2d$}\\
				&	(1-x_{i2})+(1-x_{j2}) \le 1 + c_{ij2}		&\forall e_{ij} \in CE   	\label{2e}\tag{$2e$}\\
				&	c_{ij1} + c_{ij2} \le 1 + c_{ij} 			&\forall e_{ij} \in CE		\label{2f}\tag{$2f$}\\
				&	x_{i1} - x_{j1} \le s_{ij1}            			&\forall e_{ij} \in SE  		\label{2g}\tag{$2g$}\\
				&	x_{j1} - x_{i1} \le s_{ij1}				&\forall e_{ij} \in SE  		\label{2h}\tag{$2h$}\\
				&	x_{i2} - x_{j2} \le s_{ij2}            			&\forall e_{ij} \in SE  		\label{2i}	\tag{$2i$}\\
				&	x_{j2} - x_{i2} \le s_{ij2}				&\forall e_{ij} \in SE  		\label{2j}	\tag{$2j$}\\
				&	s_{ij} \ge s_{ij1}, s_{ij} \ge s_{ij2}			&\forall e_{ij} \in SE		\label{2k}\tag{$2k$}
\end{align}

The objective function is the same as that in (\ref{eq:math}), which minimizes the weighted summation of the conflict number and the stitch number.
Constraint (\ref{2a}) is used to limit the number of colors for each node to 3.

Constraints (\ref{2b}) to (\ref{2f}) are equivalent to constraint (\ref{1a}), where 0-1 variable $c_{ij1}$ demonstrates whether $x_{i1}$ equals to $x_{j1}$, and $c_{ij2}$ demonstrates whether $x_{i2}$ equals to $x_{j2}$.
0-1 variable $c_{ij}$ is true only if two nodes connected by conflict edge $e_{ij}$ are in the same color, e.g. both $c_{ij1}$ and $c_{ij2}$ are true.

Similarly, constraints (\ref{2g}) to (\ref{2k}) are equivalent to constraint (\ref{1b}). 0-1 variable $s_{ij1}$ demonstrates whether $x_{i1}$ is different from $x_{j1}$, and $s_{ij2}$ demonstrates whether $x_{i2}$ is different from $x_{j2}$.
Stitch $s_{ij}$ is true if either $s_{ij1}$ or $s_{ij2}$ is true.

\begin{figure}[tb]
	\centering
	\includegraphics[width=0.5\textwidth]{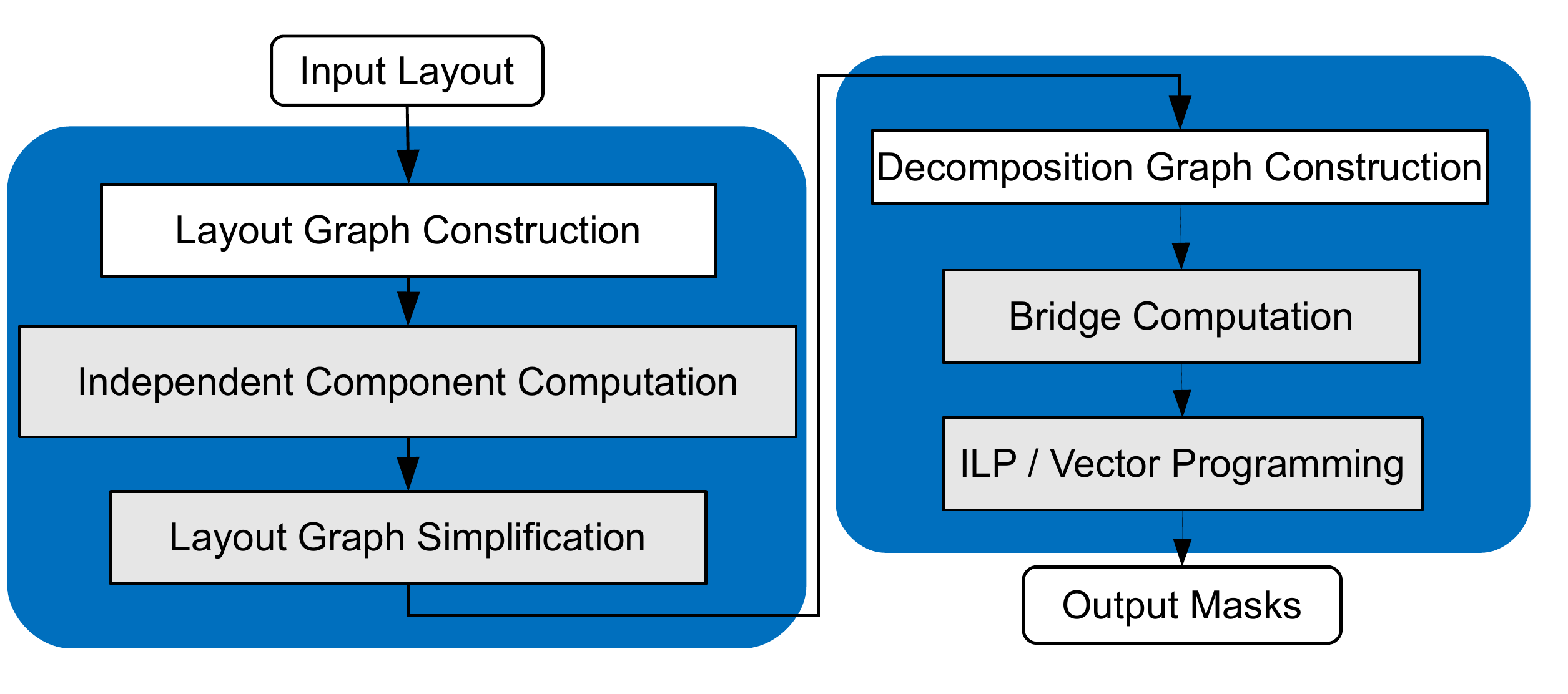}
	\caption{Basic Algorithms Flow}
	\label{fig:overview}
\end{figure}

\subsection{Acceleration Techniques}

\begin{figure*}[t]
	\centering
	\subfigure[]{\includegraphics[width=0.18\textwidth]{input}}
	\subfigure[]{\includegraphics[width=0.18\textwidth]{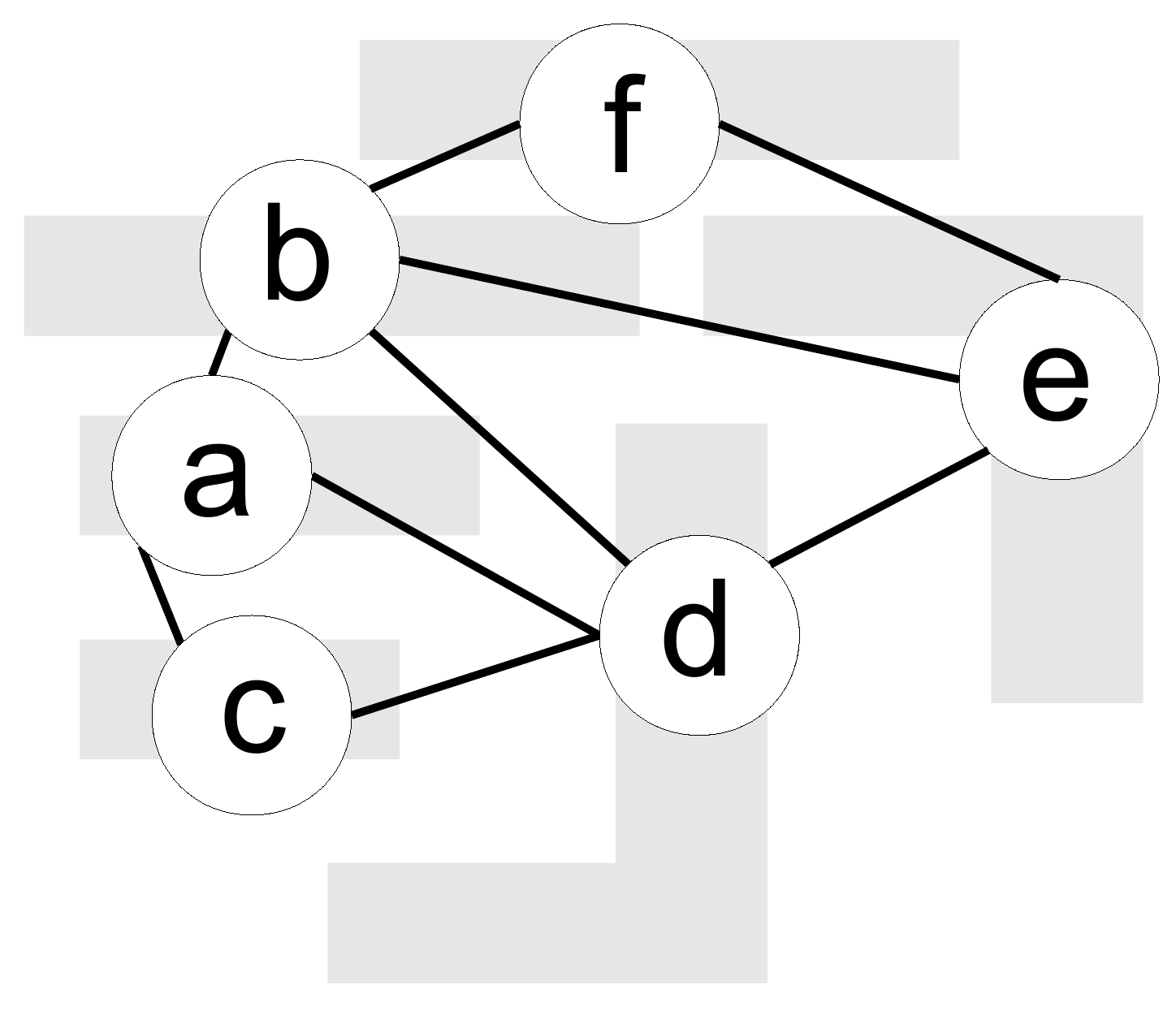}}
	\subfigure[]{\includegraphics[width=0.18\textwidth]{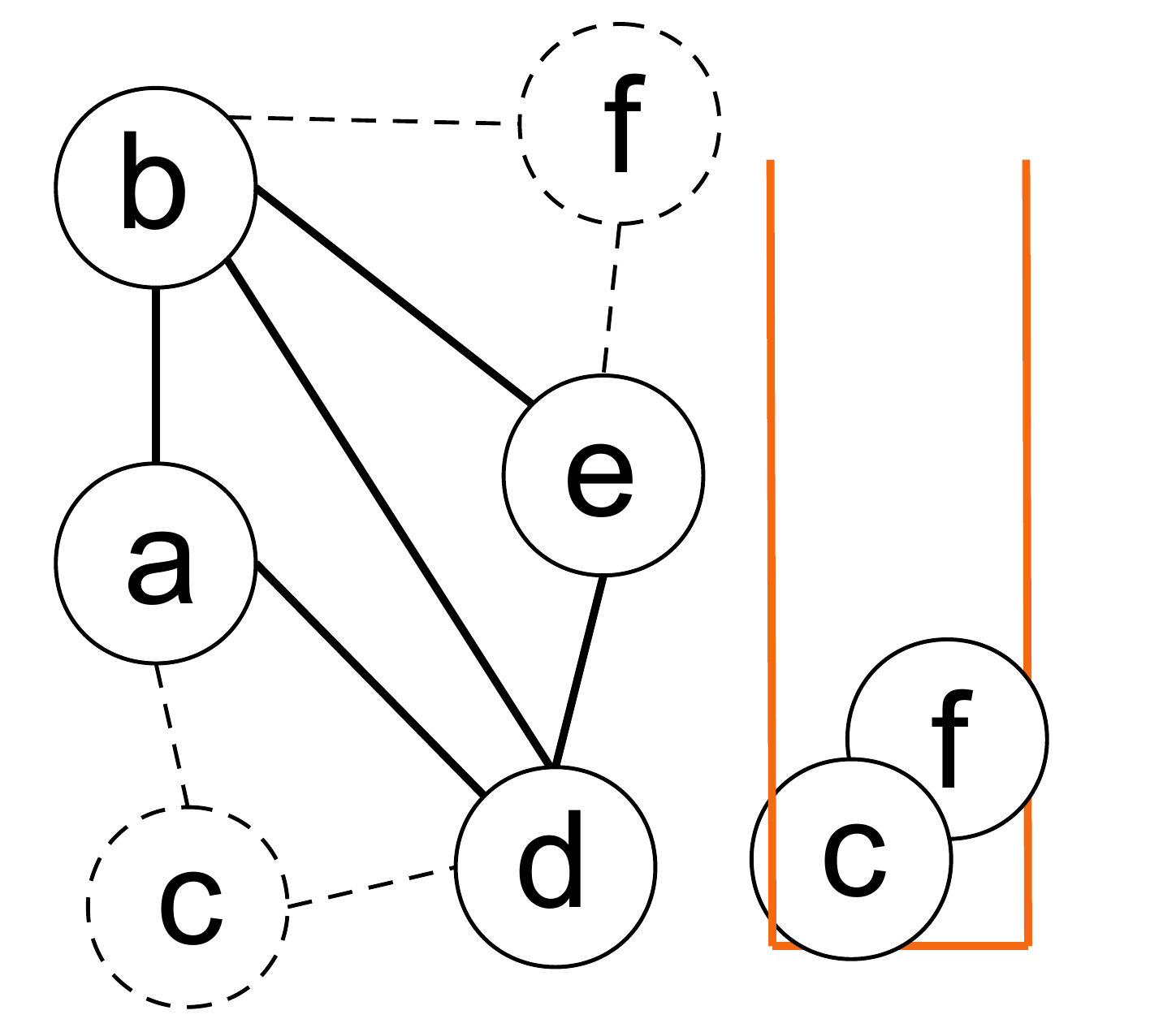}}
	\subfigure[]{\includegraphics[width=0.18\textwidth]{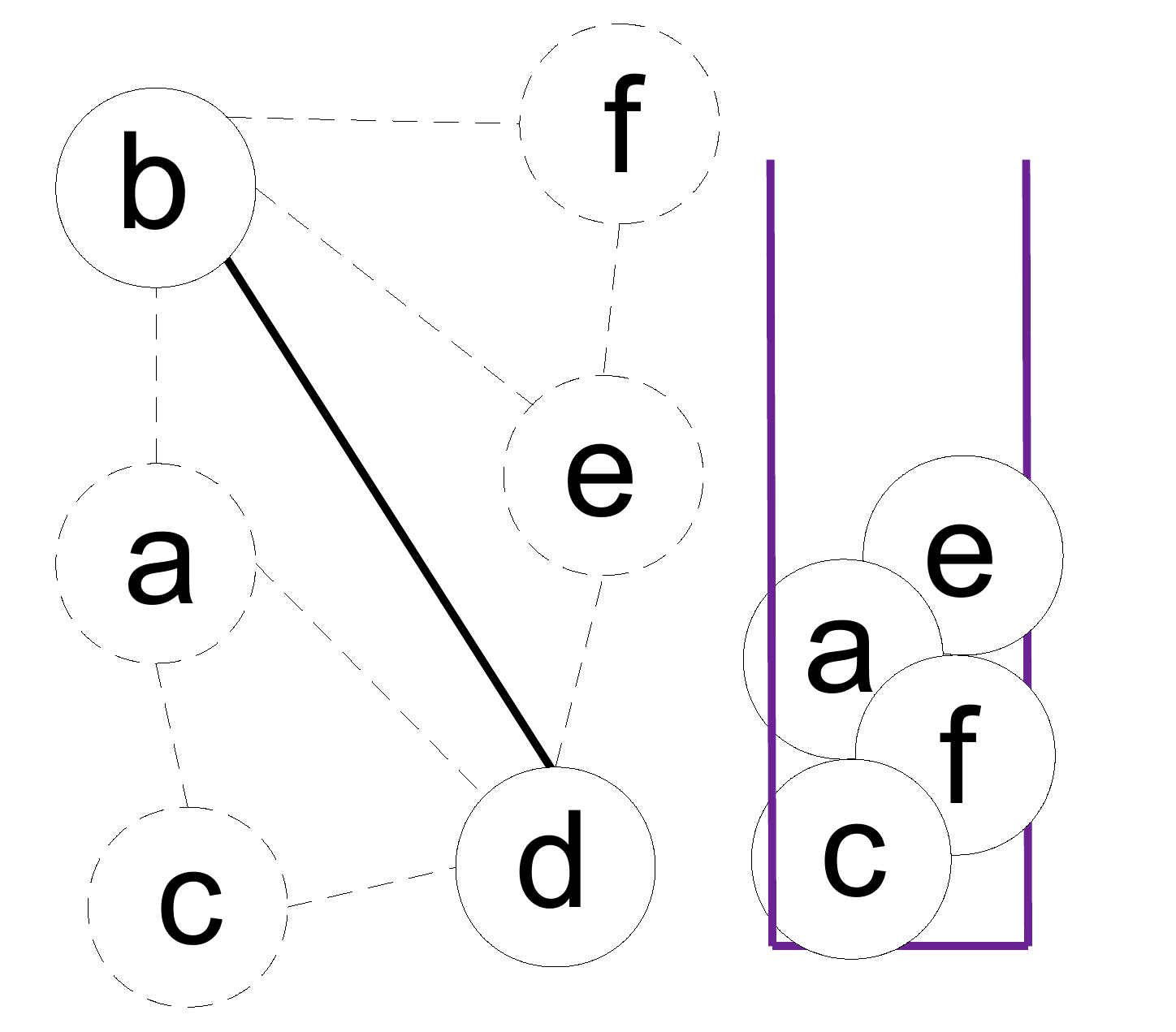}}
	\subfigure[]{\includegraphics[width=0.18\textwidth]{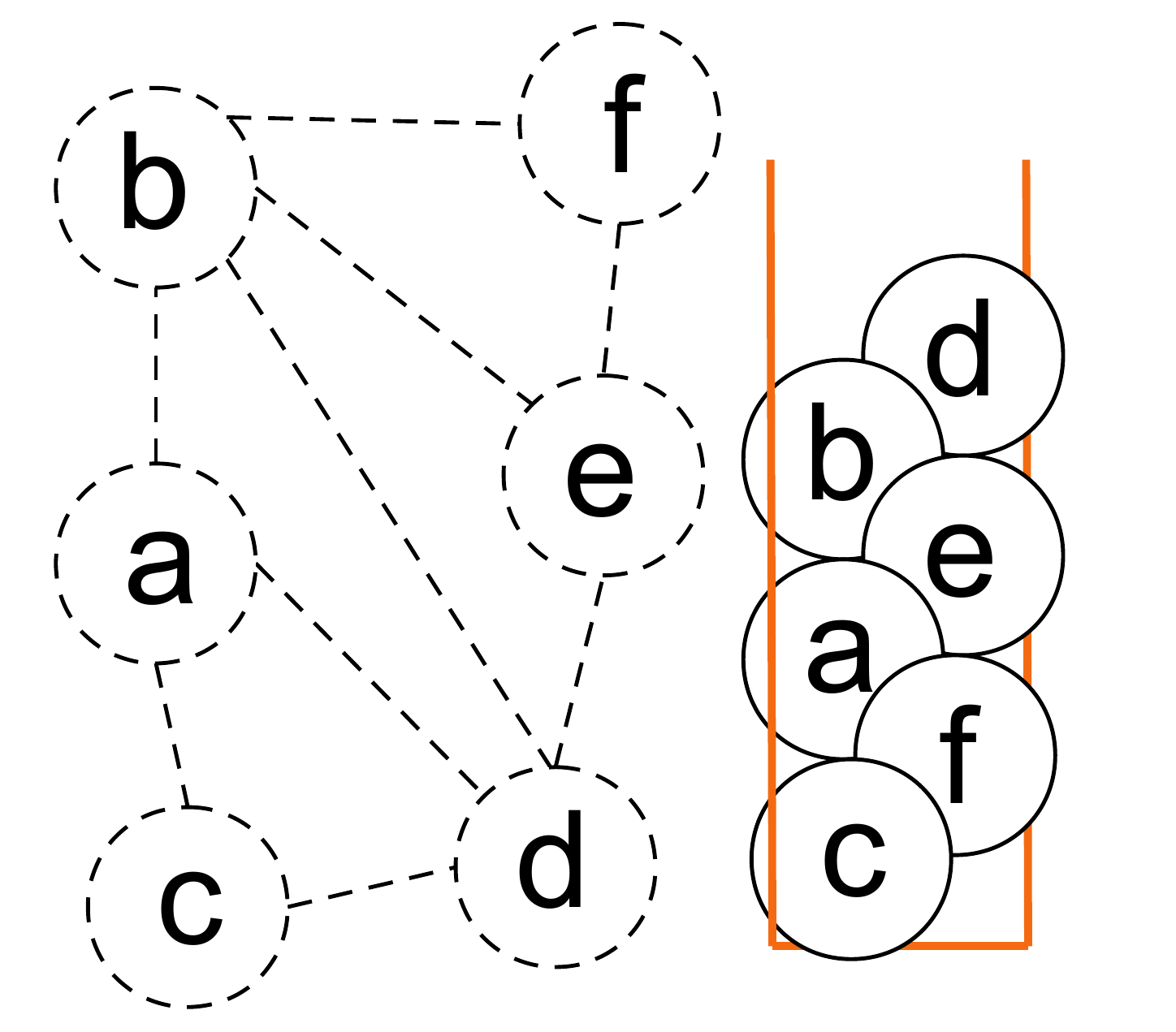}}
	\hspace{.1em}
	\subfigure[]{\includegraphics[width=0.18\textwidth]{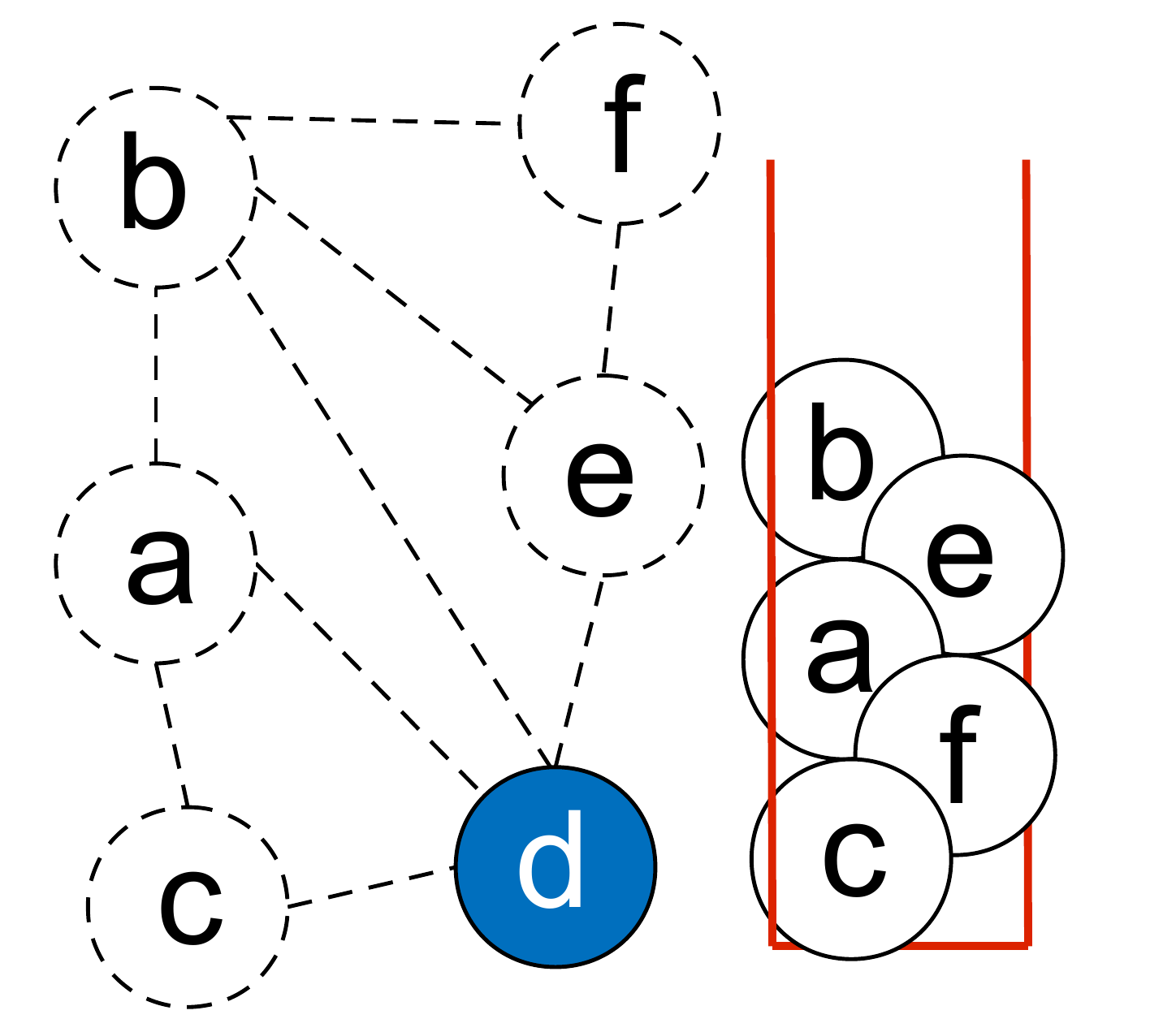}}
	\subfigure[]{\includegraphics[width=0.18\textwidth]{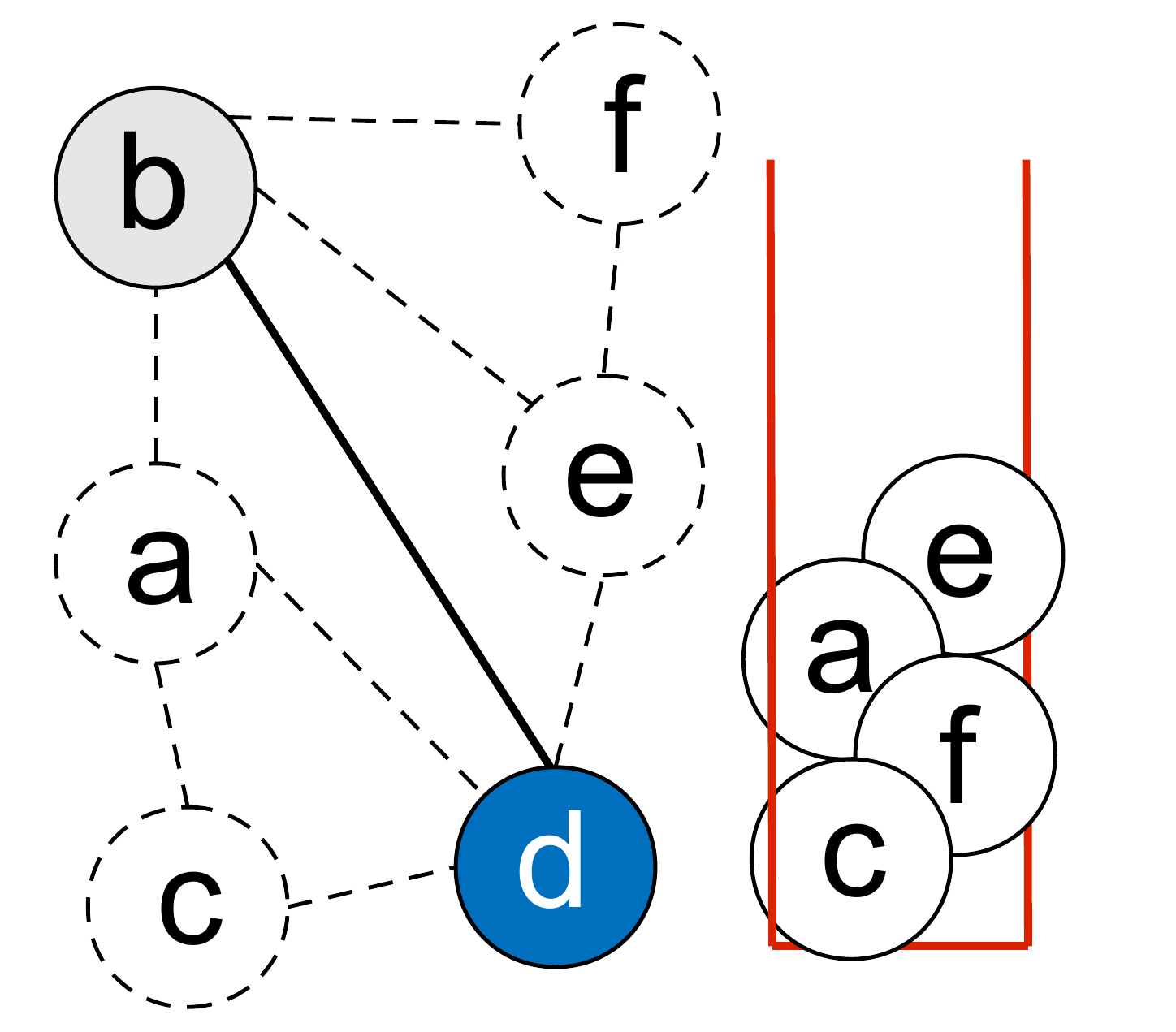}}
	\subfigure[]{\includegraphics[width=0.18\textwidth]{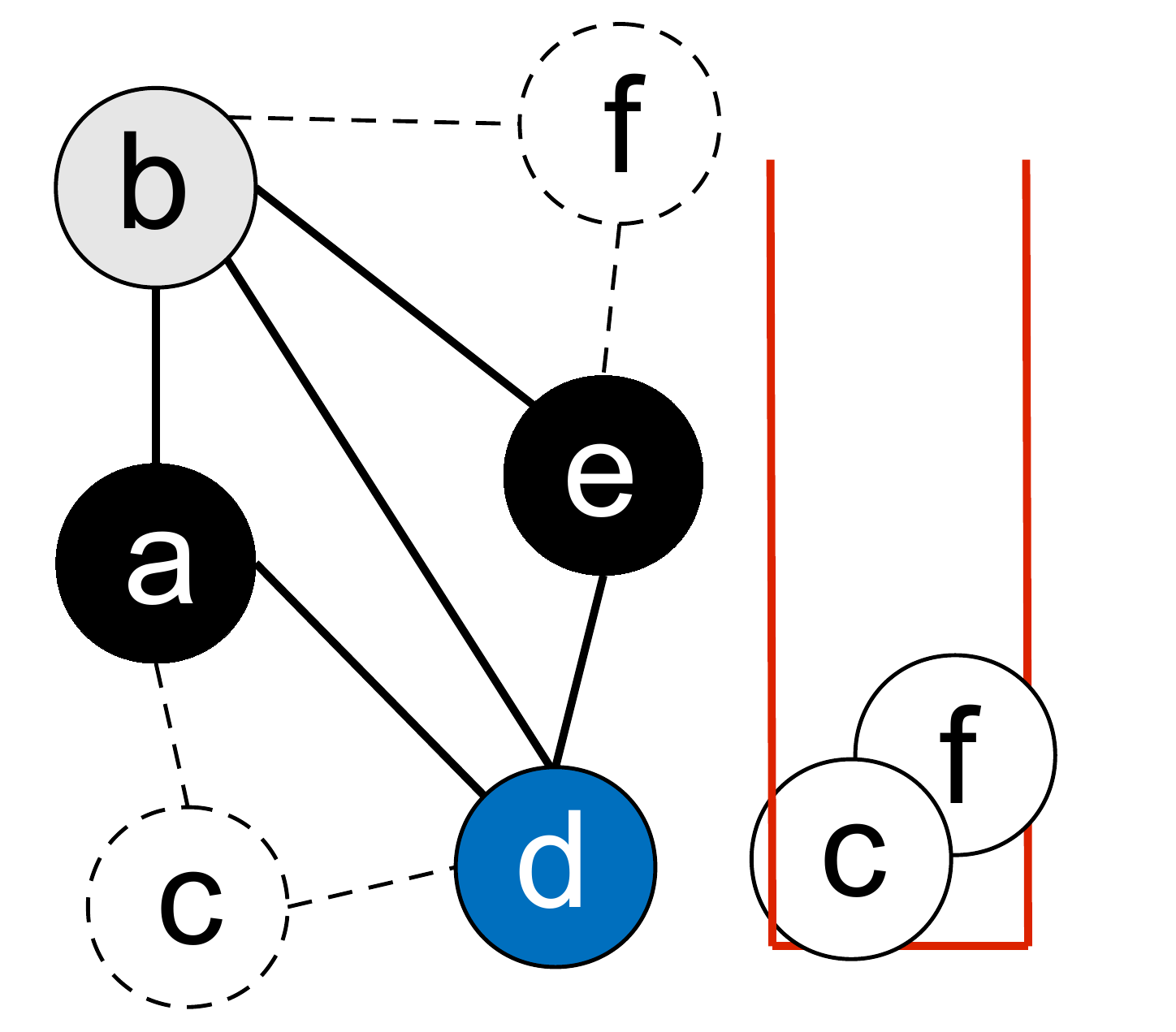}}
	\subfigure[]{\includegraphics[width=0.18\textwidth]{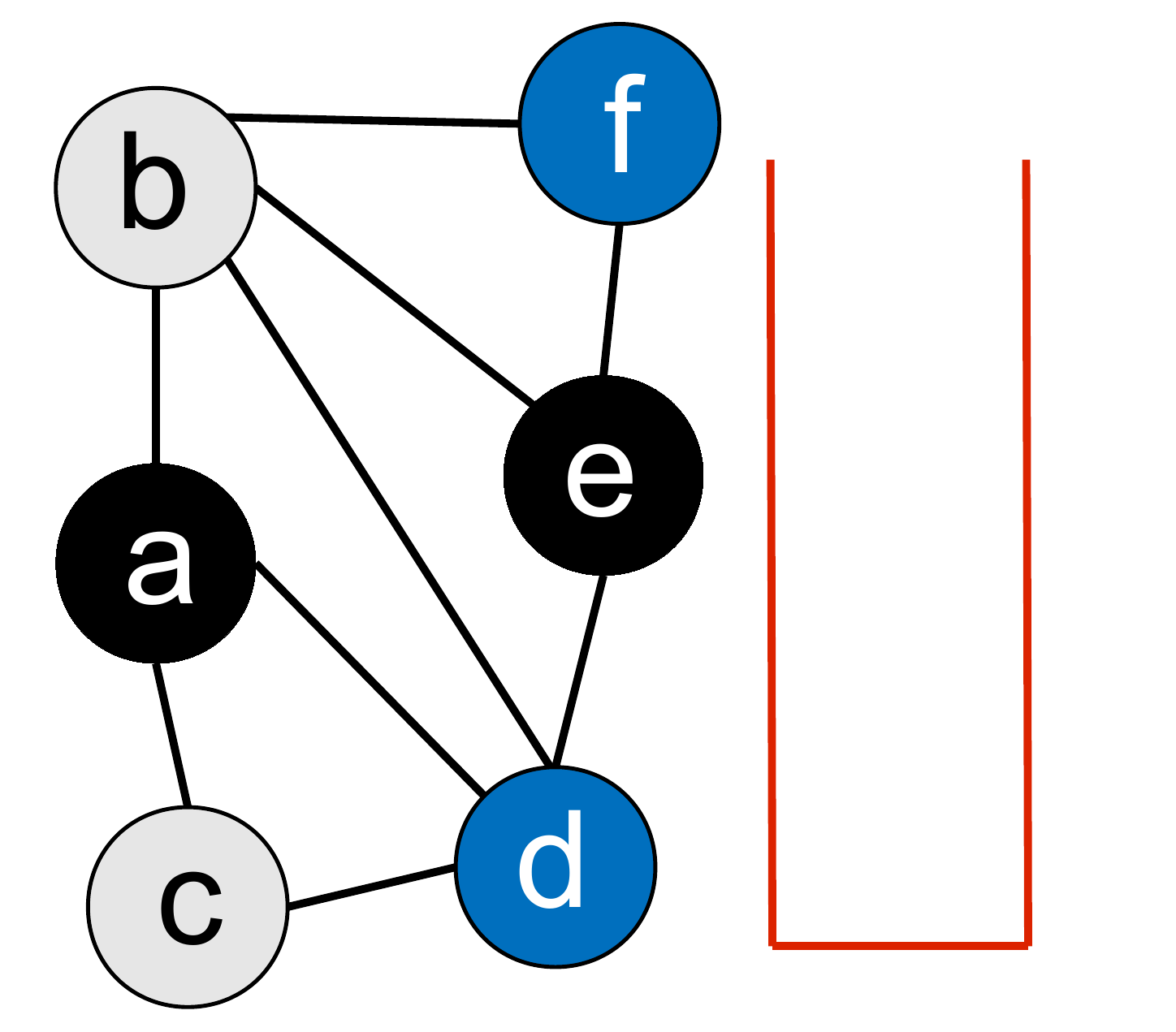}}
	\subfigure[]{\includegraphics[width=0.18\textwidth]{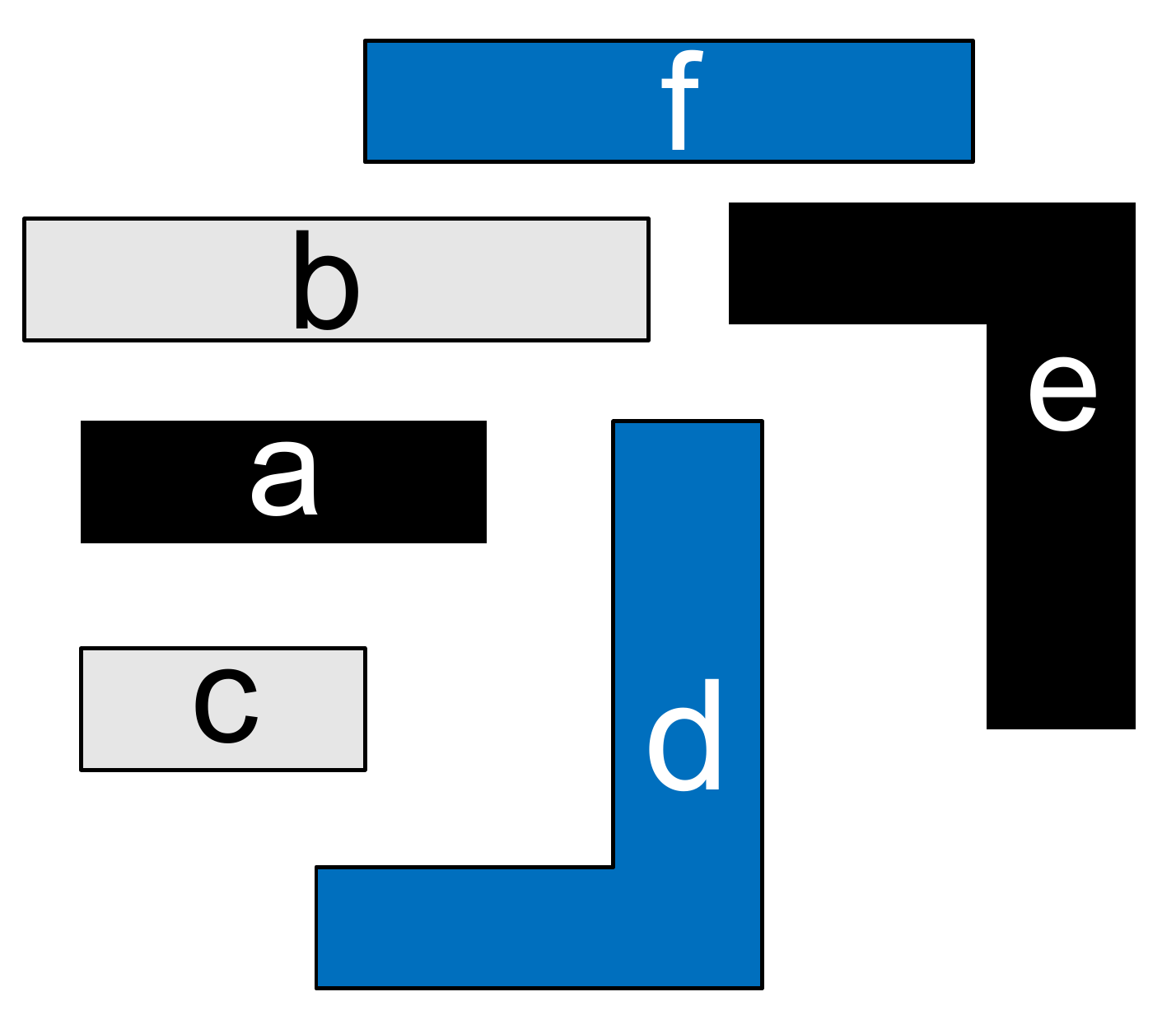}}
	\caption{This layout can be directly decomposed by layout graph simplification.
	(a) Input layout.~(b) Corresponding layout graph.
	(c)(d)(e) Iteratively remove and push in nodes with edges no more than 2.
	(f)(g)(h)(i) Iteratively pop up and recover node, and assign any legal color.~(j) Final decomposition result.
	}
	\label{fig:Simplify}
\end{figure*}

Since ILP is an NP-hard problem, its runtime increases dramatically with the size of a decomposition graph.
We propose three acceleration techniques to simplify the layout graph and the decomposition graph in order to reduce the time complexity of ILP.
As shown in Fig.\ref{fig:overview}, our acceleration flow consists of three steps: Independent Component Computation, Layout Graph Simplification and Bridges Computation.

\subsubsection{Independent Component Computation}

We propose independent component computation on the decomposition graph to reduce the ILP problem size without losing optimality.
In a layout graph of real design, we observe many isolated clusters.
Therefore, we can break down the whole design into several independent components, and apply basic ILP formulation for each one.
The overall solution can be taken as the union of all the components without affecting the global optimality.
The runtime of ILP formulation decreases dramatically with the reduction of variables and constraints, and the coloring assignment can be effectively accelerated.
Independent component computation is a well-known technique which has been applied in many previous studies \cite{DPL_ICCAD08_Kahng}\cite{DPL_ISPD09_Yuan}\cite{DPL_ASPDAC2010_Yang}\cite{DPL_ICCAD2010_Yuan}.

\subsubsection{Layout Graph Simplification}
\label{sec:Simplification}

\begin{algorithm}[htb]
\caption{Layout Graph Simplification and Color Assignment} 
\label{alg:Simplify}
\begin{algorithmic}[1]
   \REQUIRE Layout Graph $G$ to be simplified, stack $S$
   \WHILE{$\exists n \in G$ s.t. $degree(n) \le 2$}
         \STATE $S.push(n)$;
         \STATE $G.delete(n)$;
   \ENDWHILE
   \STATE Decomposition graph construction.
   \STATE TPL layout decomposition for nodes not be simplified.
   \WHILE{ $! S.empty()$}
   	\STATE $n = S.pop()$;
	\STATE $G.add(n)$;
	\STATE Assign $n$ a legal color.
   \ENDWHILE
\end{algorithmic}
\end{algorithm}

We can simplify the layout graph by removing all nodes with degree less than or equal to two.
At the beginning, all nodes with degree less than or equal to two are detected and removed temporarily from the layout graph.
This removing process will continue until all the nodes are at least degree-three.
The layout graph simplification algorithm is shown in Algorithm \ref{alg:Simplify}.

If all the nodes in the layout graph can be pushed onto the stack, Algorithm \ref{alg:Simplify} can solve TPL layout decomposition optimally in linear time. 
As an example shown in Fig.\ref{fig:Simplify}, every node can be pushed onto stack and finally be colored when is popped off.
Note that even when some nodes cannot be simplified, layout graph simplification can reduce problem size dramatically. 
Additionally, we observe that this algorithm can also partition the layout graph into several sub-graphs.


\subsubsection{Bridges Computation}

\begin{figure}[tb]
	\centering
	\subfigure[]{
		\includegraphics[width=0.12\textwidth]{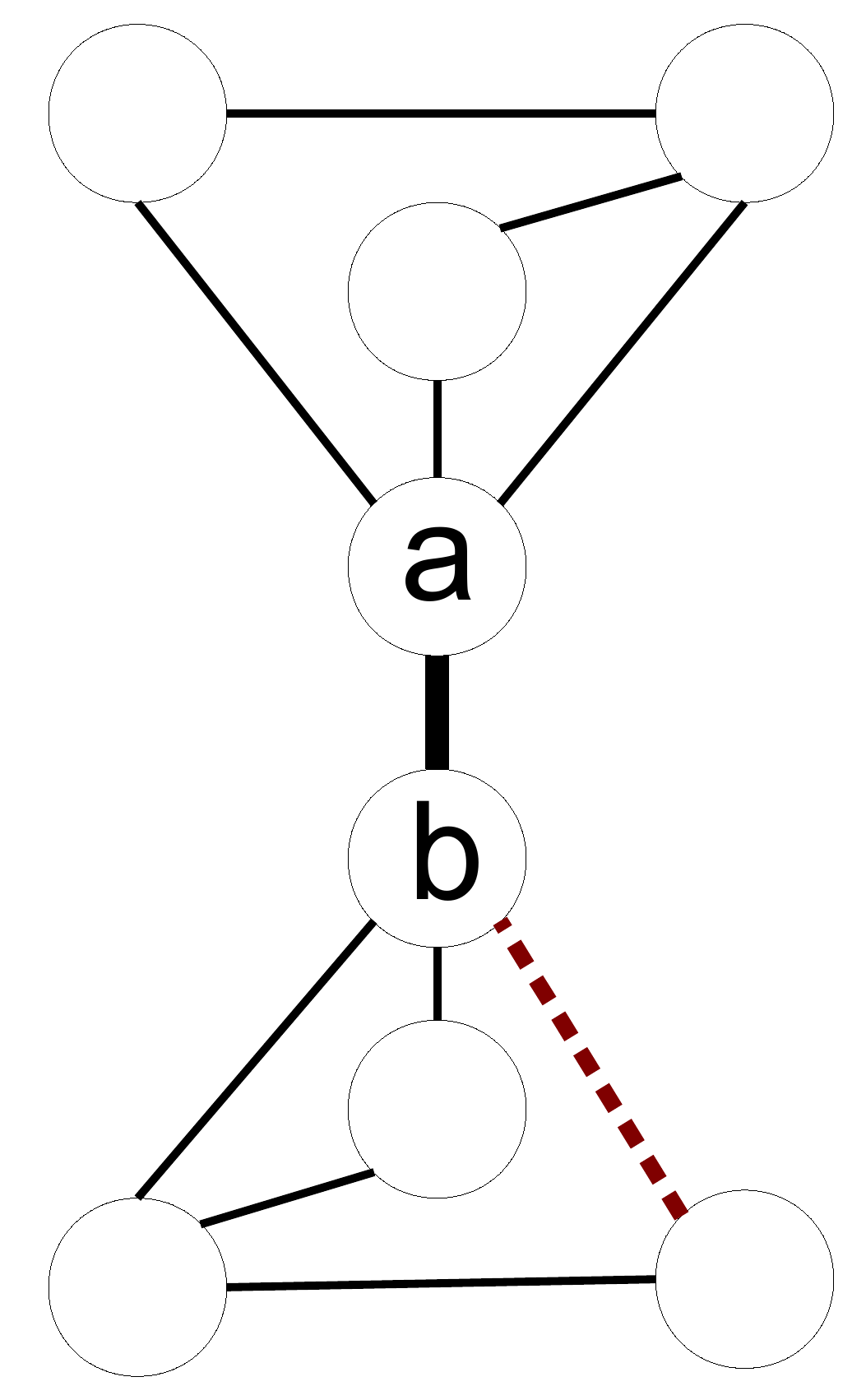}
	}
	\vspace{0.2em}
	\subfigure[]{
		\includegraphics[width=0.12\textwidth]{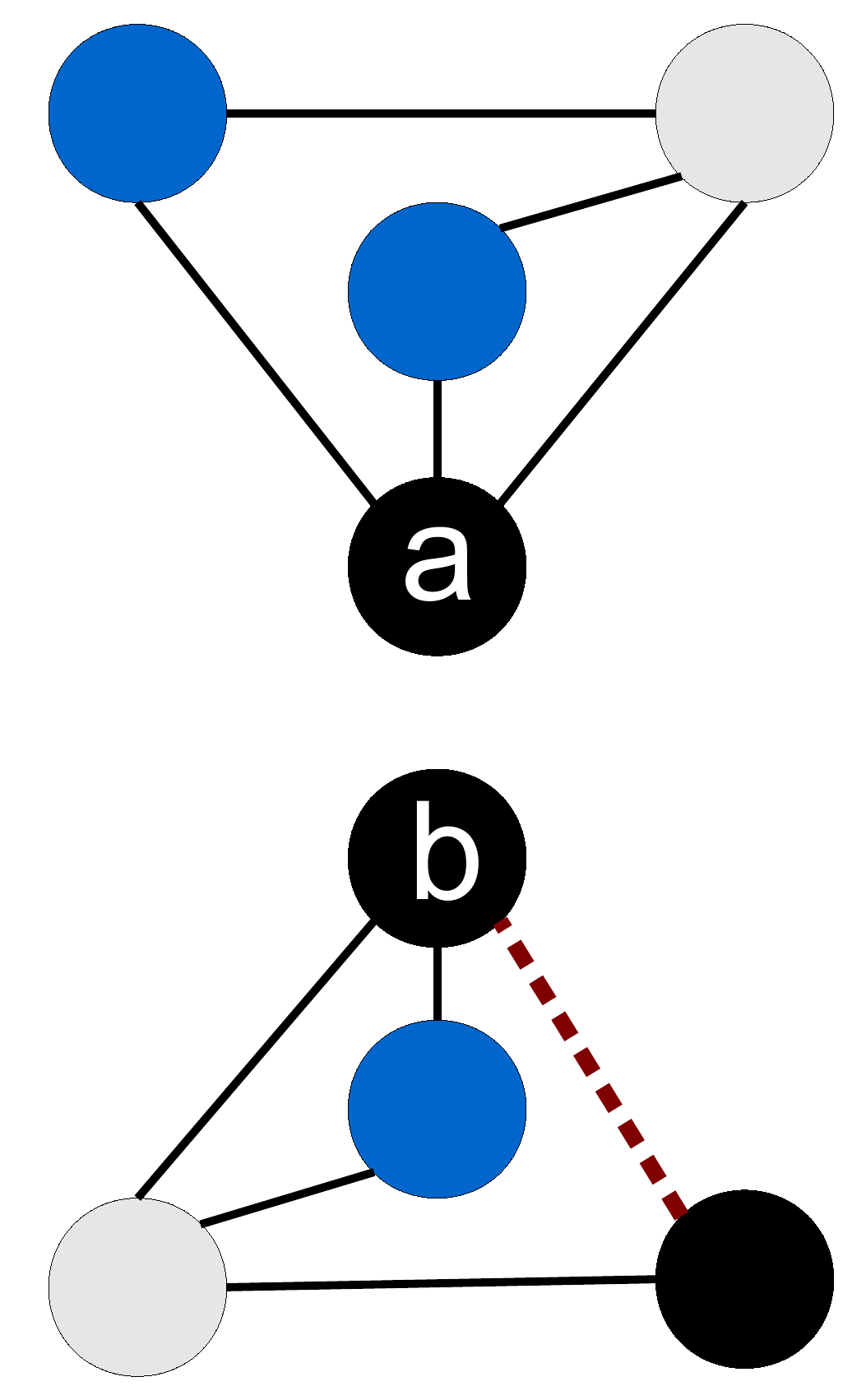}
	}
	\vspace{0.2em}
	\subfigure[]{
		\includegraphics[width=0.12\textwidth]{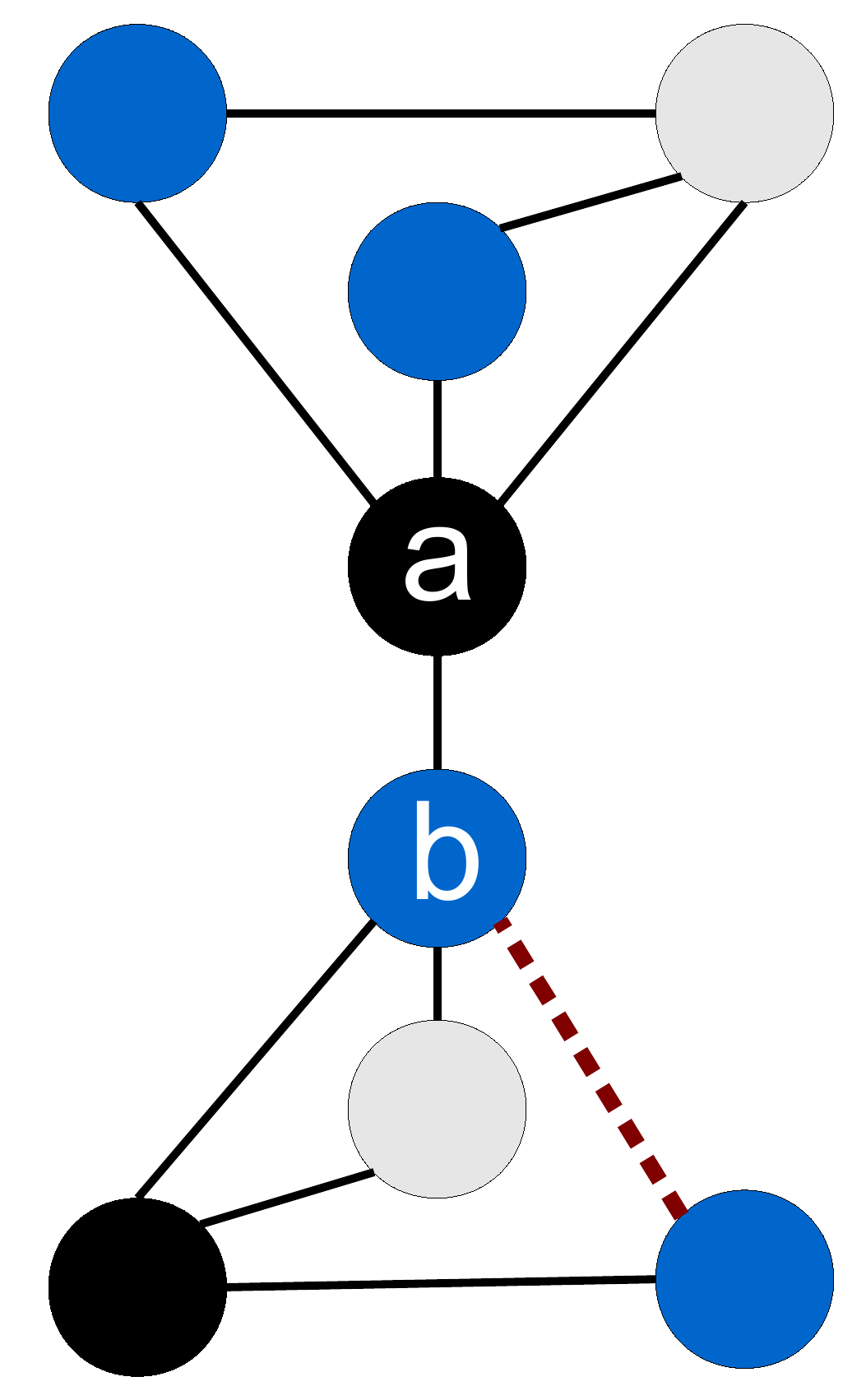}
	}
	\caption{Bridges Computation.~(a) After bridges computation, label edge $e_{ab}$ as bridge.~(b) In two sub-graphs carry out ILP formulation.~(c) Rotate colors in one sub-graph to add bridge.
	}
	\label{fig:bridge}
\end{figure}

A bridge of a graph is an edge whose removal disconnects the graph into two components.
If the two components are independent, removing the bridge can divide the whole ILP into two independent ILP formulations.

\begin{theorem}
\textbf{Partitioning decomposition graph by removing bridges does not introduce new stitches.}
\end{theorem}

An example of the bridges computation is shown in Fig. \ref{fig:bridge}.
First of all, conflict edge $e_{ab}$ is found to be a bridge. Removing the bridge divides the decomposition graph into two sides.
After ILP based color assignment, if node $a$ and node $b$ are assigned the same color,  we can rotate colors of all nodes in one sub-graph. Similar method can be adopted when bridge is a stitch edge.
We adopt an $O(|V|+|E|)$ algorithm \cite{1974Tarjan} to find bridges in decomposition graph.

Using above three acceleration techniques, the ILP formulation can still achieve optimal solutions.
In other words, our acceleration algorithms can keep optimality.
Due to page limit, we skip the detailed discussion here.

\section{SDP Based Algorithm}
\label{chap:SDP}

Although  the accelerated algorithms can simplify the problem size in many ways, ILP may still be too slow for large problems which cannot be simplified effectively.
In this section we provide an approximation algorithm to obtain more rapid solutions.
First, a novel vector programming for TPL layout decomposition is formulated.
Then we relax the vector programming into Semidefinite Programming (SDP).
Given this solution from Semidefinite Programming, we can obtain the TPL decomposition results in polynomial time.



\subsection{Vector Programming for TPL Layout Decomposition}

\begin{figure}[tb]
	\centering
	\includegraphics[width=0.23\textwidth]{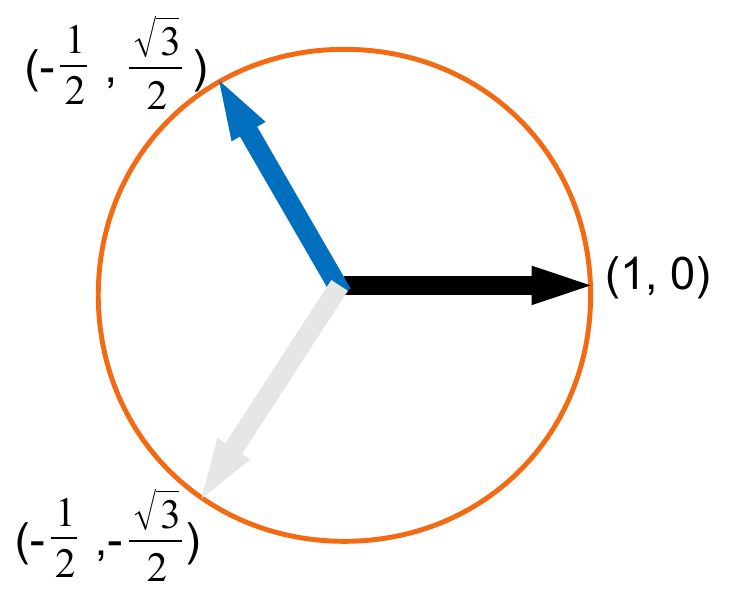}
	\caption{~Three vectors $(1, 0), (-\frac{1}{2}, \frac{\sqrt{3}}{2}), (-\frac{1}{2}, -\frac{\sqrt{3}}{2})$ represent three different colors.}
	\label{fig:vector}
\end{figure}

In TPL decomposition, there are three possible colors.
We set a unit vector $\vec{v_i}$ for every node $i$.
If $e_{ij}$ is a conflict edge, we want nodes $\vec{v_i}$ and $\vec{v_j}$ to be far apart.
If $e_{ij}$ is a stitch edge, we hope for nodes $\vec{v_i}$ and $\vec{v_j}$ to be the same.
As shown in Fig. \ref{fig:vector}, we associate all the nodes with three different unit vectors: $(1, 0), (-\frac{1}{2}, \frac{\sqrt{3}}{2})$ and $(-\frac{1}{2}, -\frac{\sqrt{3}}{2})$.
Note that the angle between any two vectors of the same color is $0$, while the angle between vectors with different colors is $2\pi/3$.

Additionally, we define the inner product of two $m$-dimension vectors $\vec{v_i}$ and $\vec{v_j}$ as follows:
\begin{displaymath}
	\vec{v_i} \cdot \vec{v_j} = \sum_{k=1}^{m} v_{ik}v_{jk}
\end{displaymath}
where each vector $\vec{v_i}$ can be represented as $(v_{i1}, v_{i2}, \dots v_{im})$.
Then for the vectors $\vec{v_i}, \vec{v_j} \in \{(1, 0), (-\frac{1}{2}, \frac{\sqrt{3}}{2}), (-\frac{1}{2}, -\frac{\sqrt{3}}{2})\}$, we have the following  property:
\begin{displaymath}
	\vec{v_i} \cdot \vec{v_j} = 
	\left\{
	\begin{array}{cc}
		1, 			& \vec{v_i} = \vec{v_j}\\
		-\frac{1}{2} 	& \vec{v_i} \ne \vec{v_j}
	\end{array}
	\right.
\end{displaymath}


Based on the above property, we can formulate the TPL layout decomposition as the following vector program \cite{approx_book2011}:

\begin{align}
    \label{eq:vp}
    \textrm{min}  & \sum_{e_{ij} \in CE} \frac{2}{3} ( \vec{v_i} \cdot \vec{v_j} + \frac{1}{2} ) + \frac{2\alpha}{3} \sum_{e_{ij} \in SE} ( 1 - \vec{v_i} \cdot \vec{v_j} ) \\
    \textrm{s.t}.\ \	
    				& \vec{v_i} \in \{(1, 0), (-\frac{1}{2}, \frac{\sqrt{3}}{2}), (-\frac{1}{2}, -\frac{\sqrt{3}}{2})\} 	 \label{3a}\tag{$3a$}
\end{align}

Formula (\ref{eq:vp}) is equivalent to mathematical formula (\ref{eq:math}).
Since the TPL decomposition is NP-hard, this vector programming is also NP-hard.
In the next part, we will relax (\ref{eq:vp}) to semidefinite programming, which can be solved in polynomial time.

\subsection{Semidefinite Programming Approximation}

Constraint (\ref{3a}) requires solutions of (\ref{eq:vp}) be discrete.
After removing this constraint, we generate formula (\ref{eq:vp2}) as follows:

\begin{align}
    \label{eq:vp2}
    \textrm{min}  & \sum_{e_{ij} \in CE} \frac{2}{3} ( \vec{y_i} \cdot \vec{y_j} + \frac{1}{2} ) + \frac{2\alpha}{3} \sum_{e_{ij} \in SE} ( 1 - \vec{y_i} \cdot \vec{y_j} ) \\
    \textrm{s.t}.\ \	& \vec{y_i} \cdot \vec{y_i} = 1 , \ \ \                   		\forall i \in V			&\label{vp2:a}\tag{$4a$}\\
    				& \vec{y_i} \cdot \vec{y_j} \ge -\frac{1}{2},   \ \ \ \forall e_{ij} \in CE 			& \label{vp2:b}\tag{$4b$}
\end{align}
This formula is a relaxation of (\ref{eq:vp}) since we can take any feasible solution $\vec{v_i} = (v_{i1}, v_{i2})$ to produce a feasible solution of (\ref{eq:vp2}) by setting $\vec{y_i} = (v_{i1}, v_{i2}, 0, 0, \cdots, 0)$, i.e. $\vec{y_i} \cdot \vec{y_j} = 1$ and $\vec{y_i} \cdot \vec{y_j} = \vec{v_i} \cdot \vec{v_j}$ in this solution. Thus if $Z_{R}$ is the value of an optimal solution of formula (\ref{eq:vp2}) and $OPT$ is an optimal value of formula (\ref{eq:vp}), it must satisfy: $Z_{R} \le OPT$.
In other words, solution of (\ref{eq:vp2}) is an approximation to that in (\ref{eq:vp}).
Since we only care about the value of $\vec{y_i}$, program (\ref{eq:vp2}) can be further simplified by eliminating the constants in the objective function:
\begin{align}
    \label{eq:vp3}
    \textrm{min} 	& \sum_{e_{ij} \in CE} (\vec{y_i} \cdot \vec{y_j}) - \alpha \sum_{e_{ij} \in SE} (\vec{y_i} \cdot \vec{y_j}) &\\
    \textrm{s.t}.	\ \	& (4a) - (4b)		&\notag
\end{align}

Without discrete constraint (\ref{3a}), programs (\ref{eq:vp2}) and (\ref{eq:vp3}) are not NP-hard now.
To solve (\ref{eq:vp3}) in polynomial time, we will show that it is equivalent to a semidefinite programming.
Semidefinite programming (SDP) is similar to linear programming which has a linear objective function and linear constraints.
However, a square symmetric matrix of variables can be constrained to be positive semidefinite.
Although semidefinite programs are more general than linear programs, both of them can be solved in polynomial time.
Besides, the relaxation based on the semidefinite programming has better theoretical results than those based on LP \cite{SDP_1996Boyed}.


Consider the following standard semidefinite program:
\begin{align}
	\label{eq:sdp}
	\textrm{SDP:\ \ min}	& \ \ \ A \bullet X						\\
				& X_{ii} = 1, \ \ \forall i \in V						\tag{$6a$}\\
				& X_{ij} \ge -\frac{1}{2}, \ \ \forall e_{ij} \in CE		\tag{$6b$}\\
				& X \succeq 0									\label{6c}\tag{$6c$}
\end{align}
where $A \bullet X$ is the inner product between two matrices $A$ and $X$, i.e. $\sum_i \sum_j A_{ij}X_{ij}$.
Here $A_{ij}$ is the entry that lies in the $i$-th row and the $j$-th column of matrix $A$.
Constraint (\ref{6c}) means matrix $X$ should be positive semidefinite.
\begin{equation}
	A_{ij} = 
	\left\{
	\begin{array}{cc}
		1, 			& \forall e_{ij} \in CE\\
		-\alpha, 		& \forall e_{ij} \in SE\\
		0,			& \textrm{otherwise}
	\end{array}
	\right.
\end{equation}

Similarly, $X_{ij}$ is the $i$-th row and the $j$-th column entry of $X$.
Note that the solution of SDP is represented as a positive semidefinite matrix $X$, while solutions of vector programming are stored in a list of vectors.
However, we can show that they are equivalent.

\begin{lemma}
\label{lem:psd}
A symmetric matrix $X$ is positive semidefinite if and only if $X = VV^T$ for some matrix $V$. 
\end{lemma}

Given a positive semidefinite matrix $X$, using the Cholesky decomposition we can find corresponding matrix $V$ in $O(n^3)$ time.

\begin{theorem}
\textbf{The semidefinite program (\ref{eq:sdp}) and the vector program (\ref{eq:vp3}) are equivalent.}
\end{theorem}
\begin{proof}
Given solutions of (\ref{eq:vp3}) $\{\vec{v_1}, \vec{v_2}, \cdots \vec{v_m}\}$, the corresponding matrix $X$ is defined as $X_{ij}=\vec{v_i} \cdot \vec{v_j}$.
In the other direction, based on Lemma \ref{lem:psd}, given a matrix $X$ from (\ref{eq:sdp}),  we can find a matrix $V$ satisfying $X=VV^T$ by using the Cholesky decomposition.
The rows of $V$ are vectors $\{v_i\}$ that form the solutions of (\ref{eq:vp3}).
\end{proof}

\subsection{Mapping Algorithm}

Solutions of program (\ref{eq:sdp}) are continuous, while optimal solutions in (\ref{eq:vp}) are discrete.
In this subsection we map the continuous solutions into discrete ones.

In the matrix $X$ generated by SDP, if $X_{ij}$ is close to $1$, then nodes $i$ and $j$ should be in the same mask, while if $X_{ij}$ is close to $-0.5$, node $i$ and node $j$ tend to be in different masks.
Our mapping algorithm is given in Algorithm \ref{alg:Rounding}, which finds the relative relationships among the nodes and maps them into three different masks.
First some triplets are constructed and sorted to store all $X_{ij}$ information.
Then, we carry out our mapping algorithm in two steps.
In the first step, if $X_{ij}$ is close to $1$, the node $i$ and the node $j$ will be in the same mask, while if $X_{ij}$ is close to $-0.5$, they will be labeled to be in different masks.
Here the vectors UnionLevel[k] and SepaLevel[k] are some user defined parameters: UnionLevel[] are close to $1$ and SepaLevel[] are close to $-0.5$.
In the second step, we continue to union the node i and the node j with maximum $X_{ij}$ until all nodes are assigned into three masks.

We use the disjoint-set data structure to group nodes into three masks. Implemented with ``union by rank" and ``path compression", the running time per operation of disjoint-set is almost constant \cite{book90Algorithm}. 
Let $n$ be the number of nodes, and the number of triplets is $n^2$.
Sorting all the triplets requires $O(n^2logn)$. 
Since all triplets are sorted, each of them can be visited at most once.
Because the runtime of each operation can be finished almost in constant time, the complexity of Algorithm \ref{alg:Rounding} is $O(n^2logn)$.

\begin{algorithm}[htb]
\caption{Mapping Algorithm} 
\label{alg:Rounding}
\begin{algorithmic}[1]
	\STATE Solve the program (\ref{eq:sdp}), get a matrix $X$.
	\STATE Label each non-zero entry $X_{i, j}$ as a triplet $(X_{ij}, i, j)$.
	\STATE sort all $(X_{ij}, i, j)$ by $X_{ij}$.
	\FOR{$k=1$ to $R$}
		\FOR{ each triple $(X_{ij}, i, j)$}
			\IF{$X_{ij} >$ UnionLevel[k] \&\& Compatible(i, j)}
				\STATE Union(i, j);
			\ENDIF
		\ENDFOR
		\FOR{ each triple $(x_{ij}, i, j)$}
			\IF{$X_{ij} <$ SepaLevel[k]}
				\STATE Seperate(i, j);
			\ENDIF
		\ENDFOR
	\ENDFOR
	\WHILE{ Masks number $> 3$}
		\STATE Pick triple with maximum $X_{ij}$ and Compatible(i, j);
		\STATE Union (i, j);
	\ENDWHILE
\end{algorithmic}
\end{algorithm}



\subsection{An Example of the SDP Based Algorithm}

\begin{figure}[tb]
	\centering
	\subfigure[]{\includegraphics[width=0.16\textwidth]{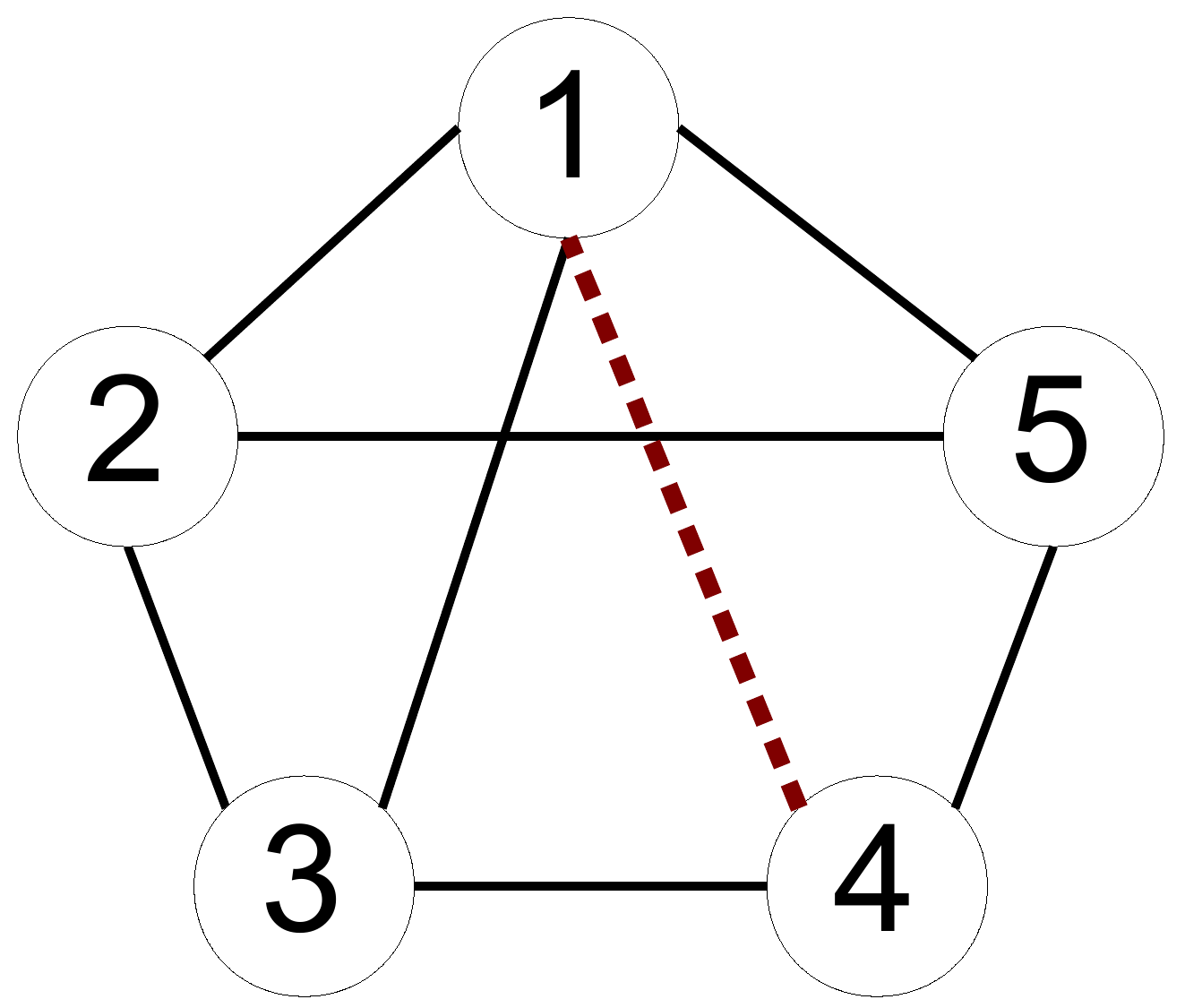}}
	\subfigure[]{\includegraphics[width=0.16\textwidth]{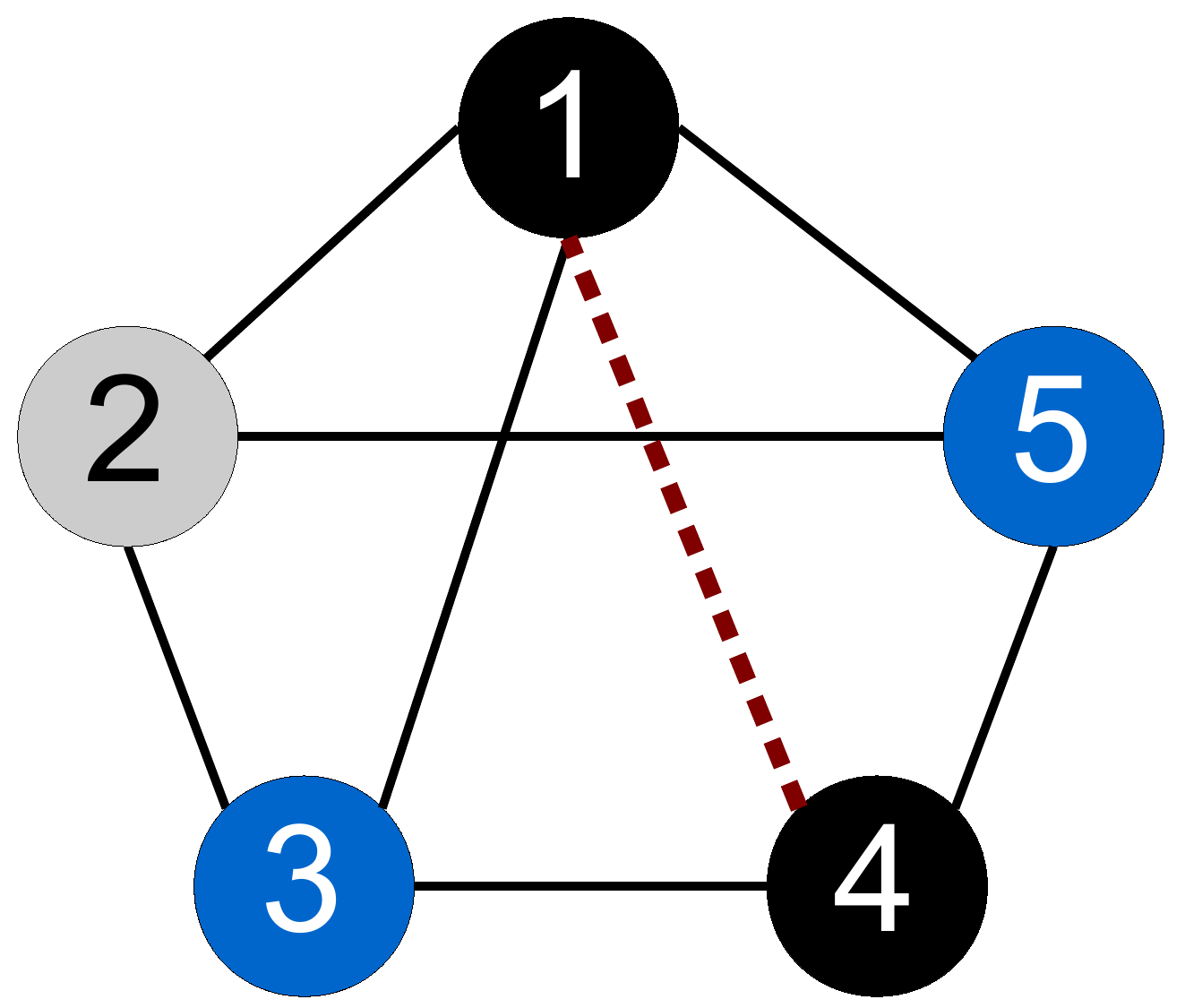}}
	\caption{Example to color decomposition graph.~(a)Input decomposition graph.
	(b)Using semidefinite programming and Algorithm \ref{alg:Rounding}, we assign nodes into 3 different colors (masks).}
	\label{fig:example}
\end{figure}

Fig. \ref{fig:example} shows an example of the decomposition graph.
This graph includes 7 conflict edges and 1 stitch edge, and can be colored with 3 colors.
Moreover, the graph is not 2-colorable since it contains odd cycles.
Here we show how the semidefinite programming can be used to solve TPL layout decomposition problem.

If we set $\alpha = 0.1$, then matrix $A$ is as follow:
\begin{displaymath}
	A = 
	\left(
		\begin{array}{ccccc}
			0 	& 1 	& 1 	& -0.1 	& 1\\
			1 	& 0 	& 1 	& 0 		& 1\\
			1	& 1 	& 0 	& 1 		& 0\\
			-0.1 & 0 	& 1 	& 0 		& 1\\
			1 	& 1 	& 0 	& 1 		& 0
		\end{array}
	\right)
\end{displaymath}

After solving the semidefinite programming (\ref{eq:sdp}), we can get a matrix $X$ as following:
\begin{displaymath}
	X =
	\left(
		\begin{array}{ccccc}
			1.0	& -0.5	& -0.5	& 1.0	& -0.5\\
				& 1.0	& -0.5	& -0.5	& -0.5\\
				&		& 1.0	& -0.5	& 1.0\\
				& \ldots&		& 1.0	& -0.5\\
				&		&		&		& 1.0\\
		\end{array}
	\right)
\end{displaymath}
here we only show the upper part of the matrix $X$.

From the matrix $X$ we can find that node 1 and node 4 should be in the same color (because $X_{14}=1.0$), and node 3 and node 5 should also be in the same color (because $X_{35}=1.0$).
However, since $X_{12}, X_{13}$ and $X_{15}$ are close to $-0.5$, nodes 2, 3 and 5 cannot be assigned in the same color as node 1.
Using Algorithm \ref{alg:Rounding}, we can map all the nodes into three colors: $\{1,4\}, \{2\}$ and $\{3, 5\}$.
The final mapping result is shown in Fig. \ref{fig:example}(b).

\section{Experimental Results}
\label{chap:result}

We implement our algorithm in C++ and test it on an Intel Core 3.0GHz Linux machine with 32G RAM.
OpenAccess2.2 \cite{OpenAccess} is used for interfacing with GDSII directly.
We choose CBC \cite{cbc} as our solver for the integer linear programming, and CSDP \cite{CSDP} as the solver for the semidefinite programming.

ISCAS-85 \& 89 benchmarks are scaled down and modified to reflect the $16$nm technology node.
The metal one layer is used for experimental purposes, because it is one of the most complex layers in terms of layout decomposition.
The minimum width and spacing become $25$nm and $30$nm.
The minimum colorable distance is set as $85$nm, and the minimum overlapping margin for stitch insertion is $10$nm.
The parameter $\alpha$ is set as $0.1$.


\subsection{Comparison}

First, we show the effectiveness of the layout graph simplification and the bridges computation.
Table \ref{tab:result1} compares the Normal ILP and the Accelerated ILP, where the Normal ILP only uses independent component computation technique, while the Accelerated ILP uses all three acceleration techniques.
Columns ``SE\#" and ``CE\#" denote the stitch edge number and conflict edge number respectively.
From Table \ref{tab:result1} we can see that layout graph simplification and bridges computation are quite effective:
the stitch edge number can be reduced by 90\%, while the conflict number can be reduced by 93\%.
The columns ``st\#" and  ``cn\#" show the stitch number and the conflict number in the final design.
``CPU(s)" is computational time in seconds.
Compared with the Normal ILP, the Accelerated ILP can achieve the same results in around 18\% of the runtime.
Note that if no accelerative technique is used, the runtime for ILP is unacceptable even for small circuits like C432. 
From Table \ref{tab:result1} we can see that both ILP formulations can achieve the optimal solutions, because of the same conflict number and stitch number. 

Second, we verify the quality and efficiency of our approximation algorithm based on semidefinite programming.
Table \ref{tab:result1} also compares the Accelerated ILP and the SDP based algorithm, where ``SDP Based" denotes the semidefinite programming based algorithm.
Note that these two methods share the same decomposition graph, i.e. both the stitch edge number and the conflict edge number in their decomposition graph are equal.
As we can see, using SDP based method the runtime can be further reduced by 42\% and the stitch number can be reduced by 7\%.
The tradeoff for this acceleration is the 9\% more conflicts.

\begin{table*}[tb]
\centering
\caption{Runtime and Performance Comparisons }
\label{tab:result1}
\begin{tabular}{|c|c|c|c|c|c|c|c|c|c|c|c|c|c|c|}
	\hline \hline
	Circuit	& Comp\#	& \multicolumn{5}{|c}{Normal ILP}				& \multicolumn{5}{|c|}{Accelerated ILP}	& \multicolumn{3}{|c|}{SDP Based}\\
	\cline{3-15}
			&			& SE\#	&CE\#	& st\#	& cn\#	& CPU(s)		& SE\#	&CE\#	& st\#	& cn\#& CPU(s)	& st\#& cn\#& CPU(s)\\
	\hline
	C432	& 261	& 300&930	& 0	& 1	& 38.01		& 32&136		& 0	& 1	& 1.11		&0&1& 0.26\\
	C499	& 418	& 566&2239	& 0	& 0	& 34.15		& 171&838	& 0	& 0	& 6.34		&0&4& 1.01\\
	C880	& 516	& 844&2219	& 1	& 3	& 25.91		& 6&24		& 1	& 3	& 0.49		&1&3& 0.06\\
	C1355	& 872	& 1008&2543	& 0	& 1	& 23.27		& 2&14		& 0	& 1	& 0.12		&0&1& 0.03\\
	C1908	& 1132	& 1332&4480	& 0	& 1	& 48.66		& 1&14		& 0	& 1	& 0.10		&0&1& 0.03\\
	C2670	& 1501	& 2186&7469	& 0	& 4	& 101.98		& 6&40		& 0	& 4	& 0.50		&0&4& 0.10\\
	C3540	& 1964	& 3197&9283	& 2	& 2	& 3975.25		& 9&40		& 2	& 2	& 0.43		&2&2& 0.11\\
	C5315	& 2767	& 4644&13211	& 5	& 0	& 173.86		& 20&70		& 5	& 0	& 0.70		&5&0& 0.18\\
	C6288	& 3740	& 5319&11394	& ?	& ?	& $>$ 7200	& 259&509	& 9	& 72	& 26.05		&9&72&1.36\\
	C7552	& 4164	& 6591&19187	& ?	& ?	& $>$ 7200	& 64&180		& 10	& 6	& 2.19		&7&9&0.46\\
	S1488	& 588	& 1932&4284	& 0	& 1	& 73.37		& 31&54		& 0	& 1	& 0.50		&0&1& 0.1\\
	S38417	& 9385	&15912&40734	& ?	& ?	& $>$ 7200	& 1617&2724	& 3	& 19	& 20.56		&3&21& 9.33\\
	S35932	& 23565	&25252&71198	& ?	& ?	& $>$ 7200	& 3776&6317	& 3	& 18	& 49.87		&3&22&33.55\\
	S38584	& 24724	&28808&74968	& ?	& ?	& $>$ 7200	& 3657&6066	& 4	& 26	& 44.16		&4&26&34.52\\
	S15850	& 20881	&33134&87358	& ?	& ?	& $>$ 7200	& 3877&6504	& 6	& 34	& 49.18		&6&39&35.92\\
	\hline
	avg.		& -		& 8735&23433	& -	& -			& -	&901.8&1568.7	&2.87&12.53&13.49		&2.67&13.73&7.80\\
	ratio		& -		& 1	& 1	& -	& -	& -			& 0.10&0.07	 & 1&1	    & 1		&\textbf{0.93}&\textbf{1.09}&\textbf{0.58}\\
	\hline \hline
\end{tabular}
\end{table*}

\subsection{Efficiency }

In order to further evaluate the scalability of our SDP based method, we create four additional benchmarks (C1-C4) to test two decomposition algorithms on very dense layouts.
Table \ref{tab:result2} lists the comparison of the Speed-up ILP and the SDP based method on these very dense layouts.
As we can see, compared with the Speed-up ILP, SDP based method can reduce stitch number by 10\% while introduces 5\% more conflicts.
Furthermore, SDP based method can achieve $140\times$ speed-up.
The reason for the dramatically acceleration is that: for a low density layout, the decomposition problem can be divided into many sub-problems, and typically each sub-problem contains no more than 20 nodes. While for high density layout, there are more nodes in each sub-problem, where SDP can be much faster than ILP.


\begin{table}[thb]
\centering
\caption{Comparison on Very Dense Layouts}
\label{tab:result2}
\begin{tabular}{|c|c|c|c|c|c|c|c|c|}
	\hline \hline
	Circuit	& SE\#&CE\#	& \multicolumn{3}{|c}{Accelerated ILP}	& \multicolumn{3}{|c|}{SDP Based}\\
	\cline{4-9}
			&	&	& st\#	& cn\#	& CPU(s)		& st\#	& cn\#	& CPU(s)\\
	\hline
	C1	& 16	&247	& 1 	& 5	& 5.5		& 0 & 6	& 0.29	\\
	C2	& 38	&289& 0	& 15	& 17.32	& 0 & 16	& 0.77 	\\
	C3	& 24	&381& 0 	& 14	& 33.41 	& 0 & 15	& 0.32 	\\
	C4	& 56 &437& 9	& 32	& 203.17	& 9 & 32	& 0.49	\\
	\hline
	avg.	& -	& -	& 2.5	&16.5&64.9	& 2.25&17.3&0.468	\\
	ratio	& -	& -	& 1	& 1	& 1		& \textbf{0.9}	& \textbf{1.05} & \textbf{0.007} \\
	\hline \hline
\end{tabular}
\end{table}

\section{Conclusion}
\label{chap:conclusion}

In this paper, we propose a general integer linear programming (ILP) formulation for the TPL layout decomposition to simultaneously minimize the conflicts  and stitches.
To improve scalability, we develop three acceleration techniques without losing solution quality:
layout graph simplification, independent component computation and bridges computation. 
Furthermore, we propose a novel semidefinite programming algorithm to improve scalability for very dense layouts.
Experimental results show that our methods are very effective. Since this is the first systematic attempt on TPL layout decomposition for general layouts, we expect to see a lot of researches as TPL may be adopted by industry in the near future.

\vspace{-.1in}
\bibliographystyle{IEEEtran}
\bibliography{./Ref/DPL,./Ref/Coloring,./Ref/LayerAssignment,./Ref/Algorithm}

\end{document}